\pdfoutput=1

\documentclass{aamas2016}
\usepackage{algorithm}
\usepackage{algorithmic}
\usepackage{url}
\begin{document}

\title{Market Share Analysis with Brand Effect}

\numberofauthors{2}

\author{
\alignauthor Zhixuan Fang\\
       \affaddr{Tsinghua University}\\
       \affaddr{Beijing, China}\\
       \email{fzx13@mails.tsinghua.edu.cn}
% 2nd. author
\alignauthor Longbo Huang\\
       \affaddr{Tsinghua University}\\
       \affaddr{Beijing, China}\\
       \email{longbohuang@tsinghua.edu.cn}\\
}

\maketitle

\begin{abstract}
In this paper, we investigate the effect of brand in market competition. Specifically, we propose a variant Hotelling model where companies and customers are represented by points in an Euclidean space, with axes being product features.
%Here we focus our attention on 1 or 2 feature market, while our methods can be adapted to higher dimensional market space.
$N$ companies compete to maximize their own profits by optimally choosing their prices, while
each customer in the market, when choosing sellers, considers the sum of product price, discrepancy between product feature and his preference, and a company's brand name, which is modeled by a function of its market area of the form $-\beta\cdot\text{(Market Area)}^q$, where $\beta$ captures the brand influence and $q$ captures how market share affects the brand.
By varying the parameters $\beta$ and $q$, we derive existence results of Nash equilibrium and equilibrium market prices and shares.
In particular, we prove that pure Nash equilibrium always exists when $q=0$ for markets with either one and two dominating features, and it always exists in a single dominating feature market when market affects brand name linearly, i.e.,  $q=1$.
Moreover, we show that at equilibrium, a company's price is proportional to its market area over the competition intensity with its neighbors, a result that quantitatively reconciles the common belief of a company's  pricing power. We also study an interesting ``wipe out'' phenomenon that only appears when $q>0$, which is similar to the ``undercut'' phenomenon in the Hotelling model, where companies may suddenly lose the entire market area with a small price increment. Our results offer novel insight into market pricing and positioning under competition with brand effect.
\end{abstract}

% Note that the category section should be completed after reference to the ACM Computing Classification Scheme available at
% http://www.acm.org/about/class/1998/.

%\category{I.2.11}{Artificial Intelligence}{Distributed Artificial Intelligence, Multi-agent Systems}

%A category including the fourth, optional field follows...
%\category{D.2.8}{Software Engineering}{Metrics}[complexity measures, performance measures]

%General terms should be selected from the following 16 terms: Algorithms, Management, Measurement, Documentation, Performance, Design, Economics, Reliability, Experimentation, Security, Human Factors, Standardization, Languages, Theory, Legal Aspects, Verification.

%\terms{Economics, Theory}

%Keywords are your own choice of terms you would like the paper to be indexed by.

%\keywords{Game theory, Market share, Brand effect}

\section{Introduction}
According to Gartner's recent report \cite{Gartner2014}, Samsung suffered a significant market share drop in smartphones in year $2014$, i.e., dropped from a dominant $32.1\%$ market share in $2013$ Q3 to only $24.4\%$ in $2014$ Q3. During the same time, Apple gained a mild share growth from $12.1\%$ to $12.7\%$.
While market fluctuations are normal, somewhat surprisingly, this occurred when Samsung decreased the average selling price (ASP)  from $\$235$ to $\$209$  of its smartphones (targeting customers who prefer low-end phones) while Apple steadily increased its ASP (targeting customers who prefer high-end phones)   \cite{Thurrott} \cite{apple}.
This phenomenon appears to contradict with the common belief that  lowering selling prices is an efficient way for boosting market share. %\textcolor{red}{need to echo this story after the theorems}
Motivated by this phenomenon,  we carry out our study on market area equilibrium based on a   general model.

In our model, companies are represented by points in an Euclidean space, where each axis denotes a particular product feature. Consumers are located uniformly in the same space and their coordinates specify their feature preference.
Each company sets its product price. Consumers determine their choice of companies based on   the summation of the product price,  discrepancy between product feature and his desire, and the brand name of the company.
In this paper, we mainly focus on one-dimensional and two-dimensional Euclidean market space, i.e., one-feature or two-feature, since consumers in the real world often consider only one or two major features when  choosing commodities. For instance, the appearance and performance for smartphone customers, or as in the Hotelling model, where two-dimension setting is also enough to describe a real world's transportation.
Even so, note that many of our results can adapt to higher dimension.
Different from location-price game models such as \cite{irmen1998competition}, where location-selection is part of a strategy, our model assumes pre-determined locations (features).\footnote{Such a model approximates the situation when companies are constrained by various factors and have limited flexibility in  choosing product features. Doing so also allows us to focus on the market space structure,  pricing and brand effect, three important aspects that have not been jointly studied before in the market equilibrium context.}

Then, in order to study the effect of brand names, we propose a novel  brand effect measure,  which has the form of \emph{Brand effect $=-\beta\cdot S^q$}, where $S$ is a company's market area, or market share under normalization, $q$ denotes how the market share contributes to the brand name, and $\beta$ represents the degree to which customers consider the brand names when
making decisions.\footnote{In practice, different companies may have different $\beta$ values. Here we use the same $\beta$ value, so as to simplify the analysis without sacrificing the economic insight. }
%The brand effect influences customers by being a constituent part of product price.
%
By varying the values of $\beta$ and $q$, we model different markets where market share has different effect on the brand name and where customers can have varying respect for the brand.

Note that our model is not restrictive. Consumers' decisions are often complex and are affected by various factors. Therefore, the exact degree to which brand name affects them has been a subject of continuous investigation, e.g., \cite{aggarwal2004effects} and \cite{buil2013influence}.
However, we already know that market share is often considered a major valid reflection of the brand's value \cite{aaker1996measuring}. Thus, we use the market share, or market area to reflect the brand effect.
Our model is motivated by the classic  Hotelling model \cite{1929} \cite{d1979hotelling}, where the total price a consumer considers is the sum of product price and the distance between consumer and company, based on linear or quadratic distance function.

Under our model, we first show the existence of Nash equilibrium among $N$ companies for the zero brand effect case, i.e.,  $q=0$ in both single and dual feature market. Then, we prove the existence for the linear brand effect case, i.e., $q=1$, in single feature market.
%in the single-feature market and provide sufficient conditions for its existence in a dual-feature market.
We also derive properties of  a company's market area and pricing strategy at equilibrium. In particular,  we show that  for any company $j$, the relation between its market area and pricing at equilibrium takes the following concise form :

\begin{eqnarray}
\text{Price}_j \propto\frac{\text{Market}_j}{\text{Competition Intensity}}. \label{eq:equilibrium-eq}
\end{eqnarray}
%where competition intensity
%
Moreover, competition intensity increases as $\beta$ grows. This result gives an explicit characterization of the equilibrium price and market area.

Besides giving the description of the equilibrium, we also discover and analyze an interesting ``wipe out'' phenomenon that can occur in the market when $\beta$ exceeds certain threshold value: some companies' market area will immediately shrink to zero with an arbitrarily small price increment. This kind of ``sudden death'' phenomenon is very similar to the ``undercut phenomenon'' in the Hotelling model with linear distance function.

 The main contributions of this work are summarized as follows.
 \vspace{-.1in}
% We list our main contributions here.
\begin{itemize}
\item To the best of our knowledge, we are the first to connect brand name effect to the market space and analyze the market equilibrium in this case.
    \vspace{-0.1in}
\item We develop techniques to prove the existence of market equilibrium when agents' utility functions are piecewise continuous. This enables us to  handle the situation when neighbor competitors change dynamically.
    \vspace{-0.25in}
\item We show that one company's pricing power is proportional to its market area at market equilibrium and is inversely proportional to competition intensity.
    \vspace{-0.1in}
\item We develop the  dimensionality reduction reasoning  approach that allows us to extend the problem into higher dimension.
%,  and gain some similar results compare with 1 dimension situation.
\end{itemize}

The organization of the rest of this paper is as follow:
Section \ref{section:relevant} reviews the existing literature. %works in relevant fields.
Section \ref{section:model} presents our model and setting, and Section \ref{section:market_structure} gives the market struture analysis results. Section \ref{section:q=0} and \ref{section:q>0} have the  results for cases $q=0$ and $q=1$.  Conclusions  come at the last and proofs are in appendix.

\subsection{Notations} Bold symbol $\mathbf{x}$ denotes a vector and capital italic symbol $\mathcal{X}$ denotes a set.  $\mathbf{cl}(\mathcal{X})$  and  $\mathbf{int}(\mathcal{X})$ denote the closure and the interior of $\mathcal{X}$. We use $X=|\mathcal{X}|$ to denote its volume, e.g., if $\mathcal{S}$ denotes the set of points in a unit circle on a plane,  $S=|\mathcal{S}|$ denotes its area.
%For uniformity, particularly when $\mathcal{X}$ is the set of points in space,
%
Exclusion of a set is denoted by $\mathcal{X}/\mathcal{A}=\{x\,|\,x\in \mathcal{X}, x\notin \mathcal{A}\}$.

\section{Related work} \label{section:relevant}
Market analysis has been a subject with extensive studies.
In  Hotelling's seminal paper \cite{1929}, an one-dimensional two-seller market with linear shipping cost was considered. It is concluded
%(\textcolor{red}{do you want to say prove here??==Hotelling didn't give rigorous proof})
in the paper that companies  follow the ``Principle of Minimum Differentiation'' in  competition, which states that companies choose similar strategy in location and pricing.
Later,  \cite{d1979hotelling} pointed out that location-price equilibrium may not always exists under the linear cost model, and a quadratic distance function was used instead to assure the existence of equilibrium. The paper also proves the principle of maximal differentiation holds, which indicates that companies  tend to produce products with  opposite characteristics.

Since then a lot of works have been done to deepen our understanding about market equilibrium. \cite{1986} proved the existence of mixed strategy of this two-stage game, and \cite{Kats199589} proved the existence of subgame perfect  equilibrium in pure strategy in an one-dimensional torus market.
In \cite{veendorp1995differentiation},  the market was extended to two-dimensional and they showed that maximal differentiation in one dimension is enough for existence of equilibrium.
% \textcolor{red}{I mean what does ``was enough'' mean}.
 \cite{irmen1998competition} extended the market to multi-dimensions and showed that firms tend to maximize the difference in the dominant dimension but minimize differentiation in other dimensions.  \cite{lauga2011product} also showed similar results for a  duoplay two-dimensional market.  %considering marginal quality cost.

Besides the existence results, properties of market area have also received a lot of attentions. Boundary of market with respect to different distance costs was studied in \cite{fetter1924economic}, \cite{hyson1950economic}, and \cite{hanjoul1989advances}. Researchers also studied the brand effect in \cite{wernerfelt1991brand}, where it was divided into two types as inertial and cost-based, while \cite{villas2004consumer} and \cite{givon1978market} studied the brand choice through repeated learning and Markov process. The main difference between previous studies on brand effects and our work is that we use market area to reflects brand effect, which is by far the first to connect market space and brand name in analyzing market equilibrium, while previous works usually treat the brand name effect and market separately.
%\textbf{==== briefly mention the differences between our work and the ones on brand effect ====}
%\textcolor{red}{they have studied brand effect already? What are the differences of our work and theirs??}

\section{Model and Preliminaries} \label{section:model}
Consider an abstract market modeled by the $K$-dimensional Euclidean space $\mathcal{M}=\mathbb{R}^K$, where each axis   represents  one feature of the products in consideration. For example,  $\mathcal{M}=\mathbb{R}^2$ can represent the smartphone market, where one coordinate can be the camera quality and the other coordinate can be the CPU speed of the phone.
As another example, $\mathcal{M}$ can  represent the market in the most traditional way as in  the Hotelling model \cite{1929}, in which case coordinates denote the locations of companies and customers.

%Suppose companies' coordinates in $\mathcal{M}$ is randomly decided.

We use $\mathcal{N}'$ to denote the set of  companies in the market.
Each company's coordinates in the market are fixed and are represented by a point $\mathbf{x_{k}}\in\mathcal{M}, \forall k\in\mathcal{N}'$.
Without loss of generality, we assume that companies are placed all over the market.\footnote{That is, for all $\mathbf{x_0}\in \mathcal{M}$, $\exists r<\infty$ such that there exists at least one company in the ball $\{\mathbf{x}|\,\left\Vert\mathbf{x}-\mathbf{x_0}\right\Vert_2<r\}$ }
For simplicity, we assume that all companies produce  products different only in features we consider, with zero cost and no limit on production capacity.
Since there are infinitely many companies in the market, we only focus on some chosen companies among them.
Specifically, let $\mathcal{M}_{chosen}(B)=\{\mathbf{x}\,|\,\left\Vert \mathbf{x}\right\Vert_{\infty}\leq B/2\}$ be a cube in $\mathcal{M}$ with  edge length being $B$ and use $\mathcal{N}$ to denote  companies inside $\mathcal{M}_{chosen}(B)$, i.e.,
\[
\mathcal{N}=\{k\,|\,\mathbf{x}_{k}\in\mathcal{\mathcal{M}}_{chosen}(B),
k\in\mathcal{N}'\}
\]

Then, we focus on the companies located inside $\mathcal{M}_{chosen}(B)$\footnote{Note that there are many ways to delimit the area, we use $\mathcal{M}_{chosen}(B)$  for simplicity in presentation, since it becomes a square in 2D market.} and
we choose $B$ so that $N=|\mathcal{N}|>0$.
%\textbf{==== why need the finite requirement? ====}
%We only consider the situation when $\mathcal{N}$ is not empty, therefore $\infty>N>0$.
%
In other words, we assume an unbounded market with infinite companies, but discuss only those   inside a bounded  set. Doing so allows us to eliminate the boundary effect as in \cite{1929} while not ignoring the competition between companies reside near the boundaries.\footnote{Note that this is important. As an example, suppose  we define a city block as $\mathcal{\mathcal{M}}_{chosen}(B)$. Then, customers live near the edge  will go to the coffee shop at the next block if they find that more worthwhile. Hence, companies at the corners still face competition from those outside the bounded set.}

We denote $P_i$ the product price of a company $i$, which is the price it charges for its  product. %\textbf{==== I don't see anywhere about the notion ``mill price''??? ====}
We assume that $P_i\in[0, P_{upper}]$,  where $P_{upper}$ is some sufficiently large but finite price upper bound  for every company.
To gain accurate result on the chosen objects, in the following, we assume that the product prices of companies outside $\mathcal{M}_{chosen}(B)$ always remain constant and only
consider the competition among companies in $\mathcal{M}_{chosen}(B)$. That is,
%, and try
%we will assume that the mill prices of companies outside $\mathcal{M}_{chosen}$ will always remain constants. Thus,
\[
P_j\equiv P_{j0}, \forall\, j \in \mathcal{N}'/\mathcal{N}.
\]
%Hence we can focus on competiton of companies in $\mathcal{N}$.
We then use $\mathbf{P}=(P_i, i\in\mathcal{N})$ to denote the price vector for the companies in consideration.

The points in $\mathcal{M}$ denote the set of customers, where each coordinate specifies a customer's desired value of that feature. %  Thus, the demand of market is uniformly distributed over the market space. %We denote each customer with a point $\mathbf{x}\in\mathcal{M}$, where each coordinate specifies her desired value of the feature.
Then, each customer will choose to purchase one unit product from a company that provides maximum satisfaction.  For instance, customers often first care about the product price when choosing companies. Then, they will also take into account the personal preferences, e.g.,  smartphone buyers may prefer the ones that match his  usage requirements, while coffee buyers may prefer the shops at closer distances. Furthermore, they may also value the brand name of a company when making decisions.
To model this customer behavior, we denote $P_i\{\mathbf{x}\}$ the \emph{aggregate} price customer $\mathbf{x}$ sees from company $i$ (its components will be specified later) and assume that customers always purchase at companies that offer the lowest price.
By choosing different forms of $P_i\{\mathbf{x}\}$, one can take into account various factors.
With the aggregate prices, we define the \emph{ownership} of a customer and the \emph{market area} of a company as follows:
\newdef{definition}{Definition}
\begin{definition}
(Ownership) $\mathcal{O}'(\mathbf{x}, \mathbf{P})$
is the set of companies who offer the \emph{lowest} aggregate price at any $\mathbf{x} \in \mathcal{M}$, i.e.,
\begin{eqnarray}
\mathcal{O}'(\mathbf{x}, \mathbf{P})=\{i \,|\, i\in\mathbf{argmin}\thinspace P_{i}\{\mathbf{x}\}\}.
\end{eqnarray}
 \end{definition}
In other words,   customers always choose to purchase goods at companies that offer the lowest price.

\begin{definition} (Market area)
Company $i$'s market area is given by $S_i = |\mathcal{S}_{i}|$, where $\mathcal{S}_{i}$
is the set of customers for which $i$ is the \emph{unique} owner, i.e.,
\begin{eqnarray}
\mathcal{S}_{i}=\{\mathbf{x}\,|\,\mathcal{O}'(\mathbf{x}, \mathbf{P})=\{i\}\}. \label{eq:market}
\end{eqnarray}
 \end{definition}
 In (\ref{eq:market}) we only consider customers that strictly prefer company $i$ to others. Later we will see that this is not restrictive, as almost all customers choose only one company.
The boundary of company $i$'s market area is denoted by $\mathcal{BR}(i)=\mathbf{cl}(\mathcal{S}_{i})/\mathbf{int}(\mathcal{S}_{i})$.

\begin{definition} (Surviving companies)
A company $i$ is called \emph{surviving} if $S_{i}>0$.
 \end{definition}

 Similarly, surviving owners of $\mathbf{x}$, denoted by  $\mathcal{O}(\mathbf{x}, \mathbf{P})$,
is the set of companies who provide the lowest price among all surviving companies at $\mathbf{x}$,  i.e.,
\[
\mathcal{O}(\mathbf{x}, \mathbf{P})=\{i\,|\,i\in\mathbf{argmin}_{i\in \{j|S_j>0\}}\thinspace P_{i}\{\mathbf{x}\}\}.
\]

In this work, unless otherwise stated, we focus on the following general $P_i\{\mathbf{x}\}$ function:
\begin{eqnarray}
P_{i}\{\mathbf{x}\}=P_{i}+D(\mathbf{x},\mathbf{x_{i}})-\beta S_{i}^q,\quad\beta, q\geq0. \label{eq:price_def}
\end{eqnarray}
Here the first term in  $P_{i}\{\mathbf{x}\}$ is company $i$'s mill price. The second term $D(\mathbf{x},\mathbf{x_{i}})$ is a distance function, which captures customer $\mathbf{x}$'s disutility due to the difference between his preferred product feature profile and that offered by company $i$. The last term models the \emph{brand effect}, i.e., how much a company's brand name helps in attracting customers.
The parameter $q$ captures how brand is related to the market area (or market share), and $\beta$ denotes how much customers value the brand name when making purchase decisions.
By varying the values of $q$ and $\beta$, our model captures a wide range of  brand effect in a market.\footnote{Here we assume that the $\beta$ value is the same for all customers. }

The objective of each company $i$ is to choose the price $P_i$, so as to maximize its revenue, i.e.,
\begin{eqnarray}
W_{i}=P_{i}\cdot S_{i}.
\end{eqnarray}

%In the following, we carry out our investigation under the above model. \textcolor{red}{$\leftarrow$I don't quite understand}
For tractability, in this paper, we  restrict our attentions mainly to  the single-feature market $\mathcal{M}= \mathbb{R}$ (1D) and the dual-feature market $\mathcal{M}=\mathbb{R}^2$ (2D).%\footnote{While this is clearly an approximation, in practice, customers often decide their choice based on one or two dominating factors. \textbf{====it will be very nice if we could find a reference to support====}}
 We also use the common $L_2$ Euclidean norm as the distance measure, i.e., % This distance measure was also used in ==.
%\footnote{Our results extend to the case with $D(\mathbf{x},\mathbf{x_{i}})=a||\mathbf{x}-\mathbf{x_{i}}||_{2}^{2}$ for general $a>0$.}
\begin{eqnarray}
D(\mathbf{x},\mathbf{x_{i}})=||\mathbf{x}-\mathbf{x_{i}}||_{2}^{2}.
\end{eqnarray}
%We will  normalize $a$ to $1$ for simplicity.
Note that linear  and pure quadratic distance functions are commonly adopted for  market analysis, e.g., \cite{1929},  \cite{d1979hotelling}.  We also use $d_{ij}=D(\mathbf{x_{i}},\mathbf{x_{j}})$ to denote the distance between $i$ and $j$.

We remark here that our model is different from the classic Hotelling model \cite{1929}, where companies also choose their locations as part of their strategies. It  has been shown in \cite{d1979hotelling} that no equilibrium exists under that model, even with two companies with linear distance costs.
%(\textcolor{red}{what about quadratic cost??}).
 We instead focus on the case when company locations (features) are predetermined. This models markets where locations cannot be arbitrarily determined by companies.

We now define the following notions that will be used repeatedly.
%surviving companies, surviving owners of the market and boundary, neighbor of a company. Then comes definition of Nash equlibirum.
% and  the Nash equilibrium of the market as follows:
 %\textcolor{red}{(This is wrong if $\mathcal{S}_{i}$ is closed, should be closed/inner)}
\begin{definition}
\label{def:neighbor} (Neighbor)
The neighbors of $i$ are the surviving owners of $\mathcal{BR}(i)$, i.e.,  $\mathcal{NR}(i)=\{j\,|\,j\neq i,\thinspace\exists\,\mathbf{x}\in\mathcal{BR}(i), \,\text{s.t.}\,j\in\mathcal{O}(\mathbf{x},\mathbf{P})\}$.
\end{definition}
%We also define Nash equilibrium here.
\begin{definition} (Nash Equilibrium)
Given the company price vector $\mathbf{P^{*}}=(P^*_i, i\in\mathcal{N})=(P_{i}^{*},P_{-i}^{*})$, we say that the
market is at pure strategy Nash equilibrium if
\begin{eqnarray*}
P_{i}^{*}=\mathbf{argmax}_{P_{i}}W(P_{i},P_{-i}^{*}), \quad\,\forall\, i\in\mathcal{N}.
\end{eqnarray*}
That is, no company can improve its profit by changing its price unilaterally.
\end{definition}

Since we only study pure strategy Nash equilibrium,   we will use Nash equilibrium for short in the rest of the paper.
It is clear that the  \emph{Individual Rationality (IR)} condition always holds since $W_{i}=P_i\cdot S_i\geq0$ for all $i\in\mathcal{N}$.
%As a company can always set $P_i=0$ and achieve $W_i=0$, it is clear that at any equilibrium, the  \emph{Individual
%Rationality (IR)} condition holds, i.e., every company has a profit $W_{i}\geq0$.
Besides, when $\beta=0$,
every company can survive by setting a very small price and gaining the customers close to its position.
For general positive $\beta$ values, however,  some companies may have zero market area even when they adopt a zero price.
%even if they charge a price no matter how close to $0$.

%\begin{definition} (Profit secure)
%Market is called \emph{profit secure} if the $\beta$ value  ensures that $P_i\{\mathbf{x}_j\}>0$ for all  $\mathbf{P}\in [0, P_{upper}]^N$,  $\forall i, j\in\mathcal{N}$ and $i \neq j$.
%\end{definition}
%
%Such secure zone of $\beta$ always exists. As we stated above, the market is secure if  $\beta=0$. Since $S_i<\infty, \forall i$ under general $\beta$ values (this will be proven in Lemma \ref{lemma:2owner}), we know that market is secure  at least for $\beta \in[0, \frac{\mathrm{min}_i \|\mathbf{x}_i-\mathbf{x}_j\|^2}{\mathrm{max}_k S_k}), \forall i,k\in \mathcal{N}$ .

%\textbf{====we need to state why we need this definition. Also, do we need to put it here? Or is it only needed for the proof? ====}
%For a better understanding in the following theorems, we now define the notion of a potential competitor.
%We now also
%Due to the brand effect, companies that are not directly competing with each other can also become neighbors if
%Companies in set $\mathcal{NR}(i)$ are the direct competitors of $i$. Companies that are not neighbors of $i$ can also compete with $i$, which means  that it may become a competitor to $i$ with a small deviation of $i$'s strategy.
Companies compete directly with their neighbors on market share. Meanwhile, we define the notion of a potential competitor, which will be useful for our later analysis.
\begin{definition}\label{def:potential} (Potential competitor)
Company $j$ is called  a \emph{potential competitor} of company $i$ if it satisfies any one of the following conditions:
\begin{enumerate}
\item[(a)] If $S_j=0$,  there exists $\mathbf{x}\in \mathbf{cl}(\mathcal{S}_i)$, such that $P_j\{\mathbf{x}\}=P_i\{\mathbf{x}\}$.
\vspace{-0.1in}
\item[(b)] If $j\in\mathcal{NR}(i)$, $\frac{|\mathbf{cl}(\mathcal{S}_j)\cap\mathbf{cl}(\mathcal{S}_i)|}{|\mathcal{BR}(i)|}=0$.
\end{enumerate}
\end{definition}
%Condition 1 reflects the situation that a non-surviving company $j$  who provides the same price at its local coordinate.
In the case when Condition (a) holds,  with an increment in $P_i$, $j$ may survive and become an actual competitor of $i$. Condition (b) can appear in the 2D market where $j$ is a neighbor of $i$, yet the length of their border line is zero, i.e., their market area only intersect at one point. Thus, $i$ may start  facing a direct competition from $j$ if $P_i$ increases.
The existence of potential competitors captures the general dynamics of different markets. However, it also greatly complicates the problem and requires a very different analysis from previous work.
%This fact that competitors can dramatically change under price deviation is
%Note that under our model, with a small price change, companies can easily become neighbors of each other and directly compete. This adds great complexity to the analysis.

%At the end of preliminaries, we need to review
Last, we have the following  theorem from  \cite{glicksberg1952further} \cite{debreu1952social} \cite{fan1952fixed}, which will be used later for proving the existence of Nash equilibrium. %, which will be used for establishing our existence results later.
\newtheorem{theorem}{Theorem}
\begin{theorem}
[\cite{glicksberg1952further}\cite{debreu1952social}\cite{fan1952fixed}] For a game with compact convex strategy space, if each player $i$'s payoff function $W_i(\mathbf{P})$ is continuous in $\mathbf{P}$ and quasiconcave in $P_i$, there exists a pure strategy Nash equilibrium. \label{thm:1952}
\end{theorem}

\section{Market Area and Ownership Analysis} \label{section:market_structure}
We now present  our results for market area and ownership. Unless otherwise stated, results in this section apply to general $K$-dimensional markets.
\newdef{lemma}{Lemma}
\begin{lemma} \label{lemma:one-survive}
For any customer $\mathbf{x}$ in the market , there exists at least one surviving company who owns him, i.e., $|\mathcal{O}(\mathbf{x}, \mathbf{P})|>0$ for all $\mathbf{x}\in \mathcal{M}$.
\end{lemma}
%
%\begin{lemma} \label{lemma:one-point}
%For any company $i$ with $S_{i}=0$, let
%$\mathcal{S}'_{i}=\{\mathbf{x}\,|\,i\in\mathcal{O}'(\mathbf{x}, \mathbf{P})\}$, i.e., the set of customers for whom company $i$ offers the lowest price. Then,
%%$|\mathcal{S}'_{i}|\leq1$, i.e.,
%$\mathcal{S}'_{i}$ contains at most one point $\mathbf{x}$ in the market.
%\end{lemma}
%

Lemma \ref{lemma:one-survive} shows that $\mathcal{O}(\mathbf{x}, \mathbf{P})\subseteq\mathcal{O}'(\mathbf{x}, \mathbf{P})$.
Thus, each customer has at least one surviving owner.
%Combining it with Lemma \ref{lemma:one-point},   we  can see that  the ``bankrupted'' companies do not really take up any market area and their impact  to the market is very limited. %Lemma \ref{lemma:potential} provides the relation between $i$ and his potential competitors.

We now have the first main theorem of the paper:
\begin{theorem}\label{thm:convex}
For any surviving company $i$ in the market,  $\mathcal{S}'_{i}=\{\mathbf{x}\,|\,i\in\mathcal{O}'(\mathbf{x}, \mathbf{P})\}$ is a convex polyhedron, and  $\mathcal{S}'_{i}=\mathbf{cl}(\mathcal{S}_{i})$.
\end{theorem}
\begin{proof}
Consider the surviving companies.  We only show the case in 2D, since situations in other dimensions are similar. First note that the set
of market area where company $i$ and $j$ provide the same price is a
straight line (hyperplane in high dimension) given by:
\begin{eqnarray*}
\mathcal{L}_{i,j}&&=\{(x,y)|P_{i}+(x-x_{i})^{2}+(y-y_{i})^{2}-\beta S_{i}\\
&&\qquad-[P_{j}+(x-x_{j})^{2}+(y-y_{j})^{2}-\beta S_{j}]=0\}\\
&&=\{(x,y)|2x(x_{i}-x_{j})+2y(y_{i}-y_{j})-P_{i}+P_{j}\\
&&\qquad-\beta S_{j}+\beta S_{i}+x_{j}^{2}+y_{j}^{2}-x_{i}^{2}-x_{j}^{2}=0\}
\end{eqnarray*}
 This line divides the plane into 2 open half space $\mathcal{H}_{i}(i,j)$, $\mathcal{H}_{j}(i,j)$,
where in $\mathcal{H}_{i}(i,j)$ we have $P_{i}\{(x,y)\}<P_{j}\{(x,y)\}$
and in $\mathcal{H}_{j}(i,j)$ we have $P_{j}\{(x,y)\}<P_{i}\{(x,y)\}$.
Thus, for any point $\mathbf{x}=(x,y)\in\mathcal{H}_{i}(i,j)$, we have
$j\notin\mathcal{O}(\mathbf{x})$ and vice versa. Let $\mathcal{U}_{i}=\cap_{j\in\mathcal{N},j\neq i}\mathcal{H}_{i}(i,j)$,
since $\mathcal{U}_{i}$ is the intersection of half
spaces, $\mathcal{U}_{i}$ is a convex polygon.

Now we prove that $\mathcal{U}_{i}=\mathcal{S}_{i}$. For
any $\mathbf{x}\in\mathcal{S}_{i}$, we have $\mathcal{O}(\mathbf{x})=\{i\}$.
Hence,  $\mathbf{x}\in\mathcal{H}_{i}(i,j),\forall j\neq i$. Thus, $\mathbf{x}\in\cap_{j\in\mathcal{N},j\neq i}\mathcal{H}_{i}(i,j)=\mathcal{U}_{i}$ and
$\mathcal{S}_{i}\in\mathcal{U}_{i}$. For any $\mathbf{x}\in\mathcal{U}_{i}$,
we see that $j\notin\mathcal{O}(\mathbf{x})$ if $j\neq i$. Since $\mathcal{O}(\mathbf{x})\neq\emptyset$, $\mathcal{O}(\mathbf{x})=\{i\}$, $\mathcal{U}_{i}\in\mathcal{S}_{i}$ and $\mathcal{U}_{i}=\mathcal{S}_{i}$.

It remains to show that $\mathcal{S}'_{i}=\mathbf{cl}(\mathcal{U}_{i})$.
For any $\mathbf{x}\in\mathbf{cl}(\mathcal{U}_{i})$, we have $P_{i}\{\mathbf{x}\}\leq P_{j}\{\mathbf{x}\},\forall j\neq i$. Hence, $i\in\mathcal{O}(\mathbf{x})$ and $\mathbf{x}\in\mathcal{S}'_{i}$, which implies
$\mathbf{cl}(\mathcal{U}_{i})\in\mathcal{S}'_{i}$. Now consider any $\mathbf{x}\in\mathcal{S}'_{i}$.
Since for all $\mathbf{x}'\in\mathbf{cl}(\mathcal{U}_{i}^{C})$,
where $\mathbf{cl}(\mathcal{U}_{i}^{C})=\mathcal{M}/\mathbf{cl}(\mathcal{U}_{i})$
is the complementary set of $\mathbf{cl}(\mathcal{U}_{i})$, we have
$i\notin\mathcal{O}(\mathbf{x}')$. Thus, we must have $\mathbf{x}\in\mathbf{cl}(\mathcal{U}_{i})$ and $\mathcal{S}'_{i}\in\mathbf{cl}(\mathcal{U}_{i})$.

In conclusion, $\mathcal{S}'_{i}=\mathbf{cl}(\mathcal{U}_{i})$,
and $\mathcal{S}'_{i}=\mathbf{cl}(\mathcal{S}{}_{i})$.
\end{proof}
%%lemma 2 owners
%%%%%%%%%%%%%
%\begin{lemma}\label{lemma:2owner}
%In the 1D market,  every customer $x$ has at most 2 surviving owners.
%\end{lemma}
%
%%lemma one intverval
%%%%%%%%%%%%%%%%%%%%%%
%\begin{lemma}
%\label{lemma:one-interval}
%In the 1D market,   each surviving  company's market consists of only one bounded continuous interval.
%\end{lemma}
%From Theorem \ref{thm:convex} we can see that every customer $x$ has at most 2 surviving owners and only customers at the boundaries will have 2 surviving owners.
%So far we have investigated the market structure: the market  looks like a tessellation graph, which is divided and filled by polyhedrons.
%So far we have a clear picture of how the market will look like under quadratic distance function and exponentiate brand effect item. The  market looks like a tessellation graph, which is divided and filled by polyhedrons. The straight border line due to quadratic distance simplifys the analysis, compare to hyperbola border in case where linear distance is adopted.
From Theorem \ref{thm:convex} we know that in the 1D market, each surviving  company's market consists of only one bounded continuous interval, while in the 2D market, it consists of a polygon.

%Let us label only the surviving companies as $\mathcal{N}_S(\mathbf{P})=\{s_1, s_2, ...,s_{N_S}\}$ in increasing order of their coordinates, i.e., $x_{s_1}<x_{s_2}<...<x_{N_S}$.
%%Hence, $x_{s_1}<x_{s_2}<...<x_{s_{N_S}}$, where $N_S=|\mathcal{N}_S(\mathbf{P})|$ is the number of surviving companies.
%Note that this does not automatically guarantee that the market area of company $i$ will be next to the market area of company $i+1$. Luckily, Lemma \ref{lemma:neighbor} below  shows that for the surviving companies, the left neighbor of $s_i$ will be $s_{i-1}$ and its right neighbor is $s_{i+1}$, i.e., it turns out that companies market area are sequentially placed and ordered by their position coordinates.
%That is, surviving companies only compete with there direct neighbors w.r.t their coordinates.
%%
%Notice that this

%lemma  one by one
%%%%%%%%%%%%%%%%%

\section{Market Equilibrium when $q=0$} 	\label{section:q=0}
In this section, we present our results  for the 1D  and 2D markets with $q=0$, i.e., the brand names do not affect customers.   In this case,  the aggregate price of company $i$ simplifies to:
\begin{eqnarray}
P_{i}\{\mathbf{x}\}=P_{i}+||\mathbf{x}-\mathbf{x_{i}}||_{2}^{2}-\beta,\quad\beta\geq0. \label{eq:p-q0}
\end{eqnarray}
%Complete analysis on both existence and properties of Nash equilbirum  will be given in the following two subsections.
%discussed in 1D and 2D cases respectively.
%especially, when in 1D market we let  $d_i=x_{i+1}-x_{i}$. We also use $l_{ij}$ to denote the length of border line between $i$ and $j$.

\subsection{Single-feature (1D) market}
In this case the market is an infinite  line and $\mathcal{N}$ denotes the set of companies in some arbitrarily chosen nonempty interval $[- B/2, B/2]$, with finite number of  companies $N>0$.
For simplicity, we denote $d_i\triangleq d_{i,i+1}= x_{i+1}-x_i$.
From Theorem \ref{thm:convex}, we can use $\mathcal{S}_i=(L_i, R_i)$ to denote each surviving company $i$'s market area, where $L_i\triangleq \inf\{x: x\in \mathcal{S}_i\}$ and $R_i\triangleq\sup\{x: x\in \mathcal{S}_i\}$ are the  boundary  points.
\begin{theorem}\label{thm:1d-q0-exist}
In the 1D market, Nash equilibrium always exists when $q=0$.
\end{theorem}
Theorem \ref{thm:1d-q0-exist} is proved by showing the utility function of a company is exactly  a parabola.
Before proving the theorem, we first have the following lemma.
\begin{lemma}\label{lemma:neighbor}
In the 1D market,  let $x_i$ denote the i-th surviving company in increasing order of their coordinates.  We have  $L_{i} = R_{i-1}$ for all $i \in \mathcal{N}_S(\mathbf{P})$.
%the left neighbor of $i$ is $i-1$ andthe right neighbor of $i$ is $i+1$.
\end{lemma}

\begin{proof} (Theorem \ref{thm:1d-q0-exist})
From (\ref{eq:p-q0}) we see that when $q=0$, the brand effect factor will be the same to every companies, regardless of the $\beta$ value. Thus, we consider $\beta=0$.
%the result for $\beta=0$ applies to general $\beta\geq0$.

%Let $\beta=0$.
For any companies $i$, other companies' aggregate prices at $\mathbf{x_i}=x_i$ are always positive due to distance.
Hence, for any companies $i$, given $P_{-i}$, $i$ can always set a price $P_i$ such that $S_i>0$. Therefore, every company will survive in the market.

For company $i$, by Theorem \ref{thm:convex} and Lemma \ref{lemma:neighbor}, we know that $i$'s market area will be an interval on the line, with neighbors  always being $i-1, i+1$, which satisfy $x_{i-1}<x_i<x_{i+1}$. The right boundary of $i-1$ and $i$ can be calculated as:
\begin{eqnarray*}
P_{i}+(R_{i-1}-x_{i})^{2}&=&P_{i-1}+(R_{i-1}-x_{i-1})^{2} \\
P_{i}+(R_{i}-x_{i})^{2}&=&P_{i+1}+(R_{i}-x_{i+1})^{2}
\end{eqnarray*}
Thus,   the  utility  $W_i$ is given by:
\begin{align}
W_i= &P_i S_i = P_i(R_i-R_{i-1})\nonumber\\
=&-P_{i}^{2}(\frac{1}{2d_{i}}+\frac{1}{2d_{i-1}})\nonumber\\
&+P_{i}(\frac{P_{i+1}+x_{i+1}^{2}-x_{i}^{2}}{2d_{i}}-\frac{P_{i-1}-x_{i}^{2}+x_{i-1}^{2}}{2d_{i-1}}).
\end{align}
Hence, the profit will be a downward parabola for $P_i\in[0,\frac{\frac{P_{i+1}+x_{i+1}^{2}-x_{i}^{2}}{2d_{i}}-\frac{P_{i-1}-x_{i}^{2}+x_{i-1}^{2}}{2d_{i-1}}}{\frac{1}{2d_{i}}+\frac{1}{2d_{i-1}}}]$ and zero otherwise.

The payoff of $i$ is continuous in $\mathbf{P}$, quasiconcave in $P_i$. From theorem  \ref {thm:1952}, Nash equilibrium always exists.
\end{proof}

\begin{theorem} \label{thm:1d-q0-condition}
In the 1D market with $q=0$, when market is at equilibrium, we must have:
\begin{equation}
P_i=\frac{2d_{i}d_{i-1}}{d_i+d_{i-1}}S_i, \forall i\in \mathcal{N}.
\end{equation}
\end{theorem}

\begin{proof}
When the market is at equilibrium,  we have:
\begin{equation}\nonumber
\frac{\mathrm{d}W_{i}}{\mathrm{d}P_{i}}=S_{i}+P_{i}\frac{\mathrm{d}S_{i}}{\mathrm{d}P_{i}}=0
\end{equation}
and
\begin{equation}\nonumber
\frac{\mathrm{d}S_{i}}{\mathrm{d}P_{i}}=\frac{1}{2d_{i}}+\frac{1}{2d_{i-1}}.
\end{equation}
Combining these two equations proved the theorem.
%\begin{equation}
%\frac{P_i}{S_i}=\frac{2d_{i}d_{i-1}}{d_i+d_{i-1}}, \forall i\in \mathcal{N}
%\end{equation}
\end{proof}

Theorem \ref{thm:1d-q0-condition} tells us that a company's price is proportional to its market area at equilibrium.
The coefficient $\frac{2d_{i}d_{i-1}}{d_i+d_{i-1}}$ is a constant for one company since their locations are fixed. Denote $\frac{1}{\gamma}=\frac{2d_{i}d_{i-1}}{d_i+d_{i-1}}$, or  $\gamma=\frac{1}{2}(\frac{1}{d_{i}}+\frac{1}{d_{i-1}})$. Notice that with bigger $d_i,d_{i-1}$ values, we will have a smaller $\gamma$, which implies higher equilibrium prices with the same market area. Therefore, $\gamma$ can be viewed as competition intensity, i.e., farther distance between companies mitigates the competition and increase company profit.
This is similar to the  maximal differentiation principle \cite{d1979hotelling}, which means that companies should not choose similar positions in the market, i.e. larger $d_i$ values.
The simple form in Theorem \ref{thm:1d-q0-condition} that equilibrium price is determined by market area over competition intensity appears to match our intuition that companies with more market share or less competition in products usually have more pricing power.

%and inversely proportional to the competition intensity. Here competition intensity is represented by $\frac{1}{2}(\frac{1}{d_{i}}+\frac{1}{d_{i-1}})$, which implies that competition intensity is decided by the market structure. That is, more close the companies are, more severe the competitions will be.
%\textbf{ ==== This paragraph is bad, rewrite. state why it is hard and importatn   (done...)==== }

\subsection{Dual-feature (2D) market}
We now turn to the 2D case.  The biggest problem of analyzing 2D market is that  companies' neighbors may change during their price change (as showed in Figure\ref{fig:2D-change}), while in 1D market, company $i$'s neighbors will always be $i-1$ and $i+1$. Due to the change in neighbors, companies' utility functions will be piecewise, i.e., utility function changes everytime a neighbor comes or goes. Moreover, since companies' locations are given arbitrarily, the shape of an company's market area may be irregular, which makes the analysis more difficult.
\begin{figure}
\centerline{\includegraphics[scale=0.35]{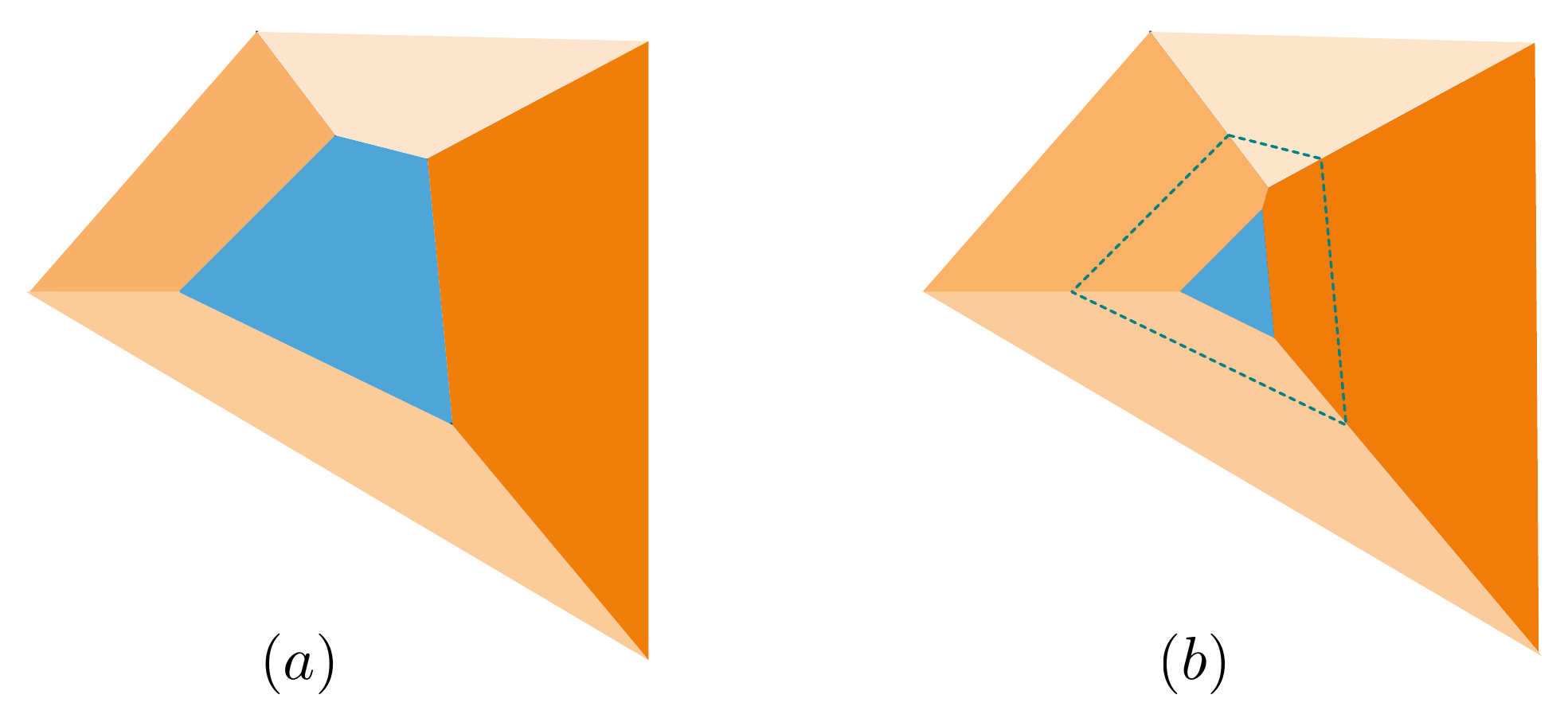}}
\caption{(a) Company $i$ (the one in the middle with blue colour) has 4 neighbors at price $P_i$;
(b)When $i$ increases its price, its market area shrinks from the dotted quadrangle to the blue triangle, and it has only 3 neighbors now. }
\label{fig:2D-change}
\end{figure}

Our method to overcome these difficulties is to transform the problem into 1D. Figure \ref{fig:2D-1D} shows such an example. We set an axis on two neighboring companies $i$ and $j$, i.e., they are both on the axis as shown in Figure \ref{fig:2D-1D} (a). Since the border line is straight and perpendicular to the axis, the position of the border line can be described in one dimension axis by variable $x_{ij}'$. The direction of the axis is determined by the way such that $x_i'<x'_{ij}<x'_j$. Note that there is only one way to determine this axis once the direction is  set, since $d_{ij}$ is constant.
\begin{figure}
\centerline{\includegraphics[scale=0.2]{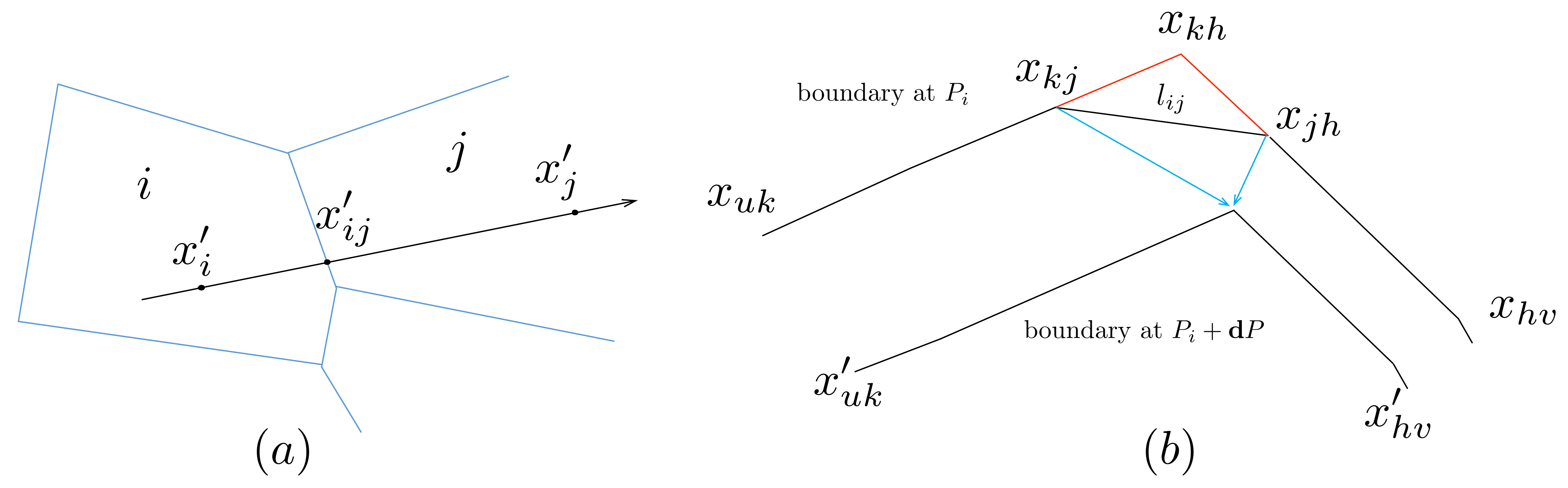}}
\caption{(a) Set an axis with $\mathbf{x}_{i}$, $\mathbf{x}_{j}$ on it so that we turn the 2D problem into 1D. Here $x'_{i}, x'_{j}$, and ${x}'_{ij}$ denote the coordinates of $\mathbf{x}_{i}$, $\mathbf{x}_{j}$ and boundary line $l_{ij}$ on the axis, respectively.
(b) When $i$ increases its price from $P_i$ to $P_i+\mathrm{d}P$, boundary between $i$ and $j$ disappears.}
\label{fig:2D-1D}
\end{figure}
In this case, we have our second main theorem in the paper.
\begin{theorem}\label{theorem:2d-exist}
In the 2D market, Nash equilibrium always exists when $q=0$.
\end{theorem}
\begin{proof}
(Sketch)
 We prove Theorem \ref{theorem:2d-exist} as follows. We first show that $S_i(P_i)$ and $W_i(P_i)$ are continuous and that $W_i(P_i)$ is  twice differentiable within each continuous piece. Then, we calculate $\frac{\mathrm{d^{2}}W_{i}}{\mathrm{d}P_{i}^{2}}$, which turns out to be downward parabola in each pieces:
 \begin{eqnarray}
\frac{\mathrm{d}^{2}W_{i}}{\mathrm{d}P_{i}^{2}}=-\sum_{j\in\mathcal{NB}(i)}\frac{6\tilde{M}_{ij}P_{i}^{2}+5\tilde{N}_{ij}P_{i}+C_{ij,1}^{2}+C_{ij,2}^{2}}{d_{ij}l_{ij}}.	\nonumber		
\end{eqnarray}
 Here $\tilde{M}_{ij}$, $\tilde{N}_{ij}$ are related to the location of $i$ and $j$, $C_{ij,1}$, $C_{ij,2}$ are related to the location of two end points of $i$ and $j$'s border, and $l_{ij}$ denotes the length of border line between $i$ and $j$.
  %(See technical report \cite{TechReport} for more details).
 We then show its monotonicity, that is, the functions are monotone within each differential interval and monotone at each discontinuous point (Fig. \ref{fig:dw2dp} shows an example). Specifically, we first show that each pieces is on the left side of its axis of symmetry and prove monotonicity inside each pieces. For the discontinuous point between two pieces, the situation is shown in Fig. \ref{fig:2D-1D} (b):  $i$ and $j$'s boundary $l_{ij}$ disappears when $i$ increases its price from $P_i$ to $P_i+\mathrm{d}P$, i.e., $i$'s boundary changes from $\overline{\mathbf{x}_{uk}\mathbf{x}_{kj}\mathbf{x}_{jh}\mathbf{x}_{hv}}$ to $\overline{\mathbf{x}_{uk}'\mathbf{x}_{kh}'\mathbf{x}_{hv}'}$. We have applied triangle inequality on the triangle $ \bigtriangleup\mathbf{x}_{kh}\mathbf{x}_{kj}\mathbf{x}_{jh}$ to show that $\frac{\mathrm{d^{2}}W_{i}}{\mathrm{d}P_{i}^{2}}$ increases after the change of neighbors.
 Lastly, we prove that company $i$'s payoff function is quasi-concave in $P_i$ and then prove the existence of equilbrium.
\begin{figure}
\centerline{\includegraphics[scale=0.2]{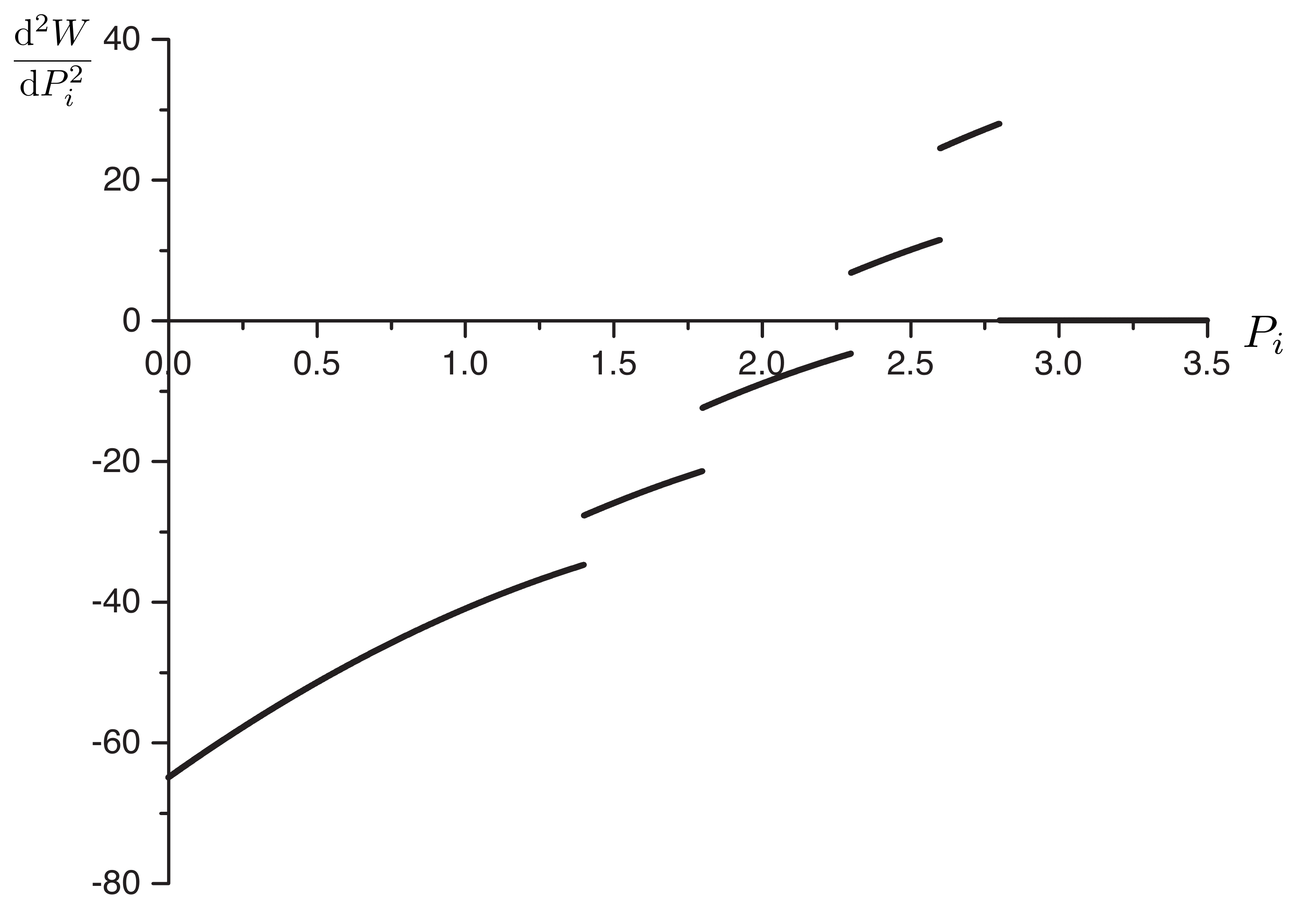}}
\caption{An example of  $\frac{\mathrm{d}^{2}W_{i}}{\mathrm{d}P_{i}^{2}}$. Each piece of  $\frac{\mathrm{d}^{2}W_{i}}{\mathrm{d}P_{i}^{2}}$ is a fragment from an increasing downward parabola, therefore  $\frac{\mathrm{d}^{2}W_{i}}{\mathrm{d}P_{i}^{2}}$ is a monotonically increasing function on $[0,P_{upper}]$.}
\label{fig:dw2dp}
\end{figure}
\end{proof}
%This is one of the main results in our paper.
This theorem guarantees the existence of  equilibrium in 2D market, regardless of the change in each company's neighbors, or the change in the shape and location of their market area.
%
%Thus, one cannot directly apply the results in Theorem \ref{thm:1d-q0-exist}, because the utility function varies every time the company changes its neighbors.
% We develop a method of analyzing this kind of piecewise utility function to prove the theorem.

%The proof for Theorem \ref{theorem:2d-exist} will be presented in Section \ref{section:proof-2dtheroem}.
\begin{theorem} \label{thm:2d-condition}
In the 2D market with $q=0$, when the market is at Nash equilibrium, we have:
\begin{equation}
P_i=\frac{1}{\sum_{j\in\mathcal{NR}(i)}\frac{l_{ij}}{2d_{ij}}}S_i, \quad\forall i\in \mathcal{N}.
\end{equation}
\end{theorem}
The factor $\gamma=\sum_{j\in\mathcal{NR}(i)}\frac{l_{ij}}{2d_{ij}}$  represents the competition intensity. For a company $i$, farther distance to competitors (bigger $d_{ij}$) can reduce the competition intensity, while longer contiguous border (smaller $l_{ij}$) increases it.

\section{Market Equilibrium when $q=1$} 	\label{section:q>0}
In this section, we discuss the situation when $q=1$, i.e.,  when the market area has a linear relationship with the brand name. We show that the interesting ``wipe out" phenomenon appears when $q>0$ (In particular,  it happens when $q=1$).
%We know that bigger $q$ will intensify the feedback: leading companies in the market can take advantage of the market area gap to drive out other competitors.
%It turns out that in this case, the interesting ``Wipe Out" phenomenon appears when $q>0$.
%%This is similar to  the ``undercut" phenomenon in Hotelling model when a company suddenly loses its entire market area.

The ``wipe out'' phenomenon  substantially increases the difficulty in analyzing the problem,  because in this case a company's market area can suddenly shrinks to zero after some threshold price. In this case,   its neighbors' utility functions are  not continuous.
This is exactly the same problem as  in the classic Hotelling model, where ``undercut'' destroys the continuity of the utility function, and therefore lead to the non-existence of equilibrium.

Since this case is very different from the previous cases due to the  ``wipe out" phenomenon, we first discuss it below before  presenting the existence and necessary conditions for equilibrium.
We focus on the 1D market and $q=1$.

\subsection{The ``Wipe-out" Phenomenon} \label{subsec:wipe-out}

%First note that  $P_{i}\in[0,P_{upper}]$ is a compact convex set.
%
Fix $P_{-i}=(P_{1},...P_{i-1},P_{i+1},...P_{N})$ and consider $P_{i}=0$.
If $S_{i}(P_{i}=0)=0$, then  $W_{i}(P_{i})=0$ for all  $P_i\in[0,P_{upper}]$ and it is concave.
% which is also concave. Otherwise
If  $S_{i}(P_{i}=0)>0$,
let us gradually increase $P_i$ from zero and consider a price $P_i\in[0, P_{upper}]$ during this process. Denote its surviving left and right neighbors by $i-1, i+1$, and denote $R_i$ as the right boundary of $i$. Recall that in this case, $P_i\{R_i\}=P_{i}+(R_{i}-x_{i})^{2}-\beta S_{i}$.  Thus, we have:
%Then, we have:
%\begin{eqnarray}
%P_{i}+(R_{i}-x_{i})^{2}-\beta S_{i}^q&=&P_{i+1}+(R_{i}-x_{i+1})^{2}-\beta S_{i+1}^q	\nonumber\\
%P_{i}+(R_{i-1}-x_{i})^{2}-\beta S_{i}^q&=&P_{k}+(R_{i-1}-x_{i-1})^{2}-\beta S_{i-1}^q.  	\label{eq:1D-LR}
%\end{eqnarray}
\begin{eqnarray}
P_i\{R_i\}&=&P_{i+1}\{R_i\}		\nonumber\\
P_i\{R_{i-1}\}&=&P_{i-1}\{R_{i-1}\}  \label{eq:1D-LR}
\end{eqnarray}
Suppose now  $i$ changes its price by a small amount $\mathrm{d}P_{i}$, i.e., $P_{i}'=P_i+\mathrm{d}P_i$, such that $i-1, i+1$ still remain its neighbors, i.e., they still survive.
Denote $\mathrm{d}R_{i}$ and $\mathrm{d}R_{i-1}$ the deviations of the
right and left boundaries of $i$.  The change of market area for $i$ can then be written as
$\mathrm{d}S_{i}=\mathrm{d}R_{i}-\mathrm{d}R_{i-1}$, and the price of $i$ at new boundary is:
\begin{eqnarray}
 P_i'\{R_i+\mathrm{d}R_i\}=P_{i}+\mathrm{d}P_{i}+(R_{i}+\mathrm{d}R_{i}-x_{i})^{2}-\beta(S_{i}+\mathrm{d}S_{i}). \nonumber
\end{eqnarray}
 Similar to (\ref{eq:1D-LR}), we have two equations for the new boundaries:
\begin{eqnarray}
P_i'\{R_i+\mathrm{d}R_i\}&=&P_{i+1}'\{R_i+\mathrm{d}R_i\}		\nonumber\\
P_i'\{R_{i-1}+\mathrm{d}R_{i-1}\}&=&P_{i-1}'\{R_{i-1}+\mathrm{d}R_{i-1}\}. \label{eq:1D-LR_deviate}
\end{eqnarray}
%\begin{eqnarray}
%&&\quad P_{i}+\mathrm{d}P_{i}+(R_{i}+\mathrm{d}R_{i}-x_{i})^{2}-\beta(S_{i}+\mathrm{d}S_{i})^q \nonumber\\
%&&=P_{i+1}+(R_{i}+\mathrm{d}R_{i}-x_{i+1})^{2}-\beta(S_{i+1}+\mathrm{d}S_{i+1})^q		\nonumber\\
%&&\quad P_{i}+\mathrm{d}P_{i}+(R_{i-1}+\mathrm{d}R_{i-1}-x_{i})^{2}-\beta(S_{i}+\mathrm{d}S_{i})^q \nonumber\\
%&&=P_{i-1}+(R_{i-1}+\mathrm{d}R_{i-1}-x_{k})^{2}-\beta(S_{i-1}+\mathrm{d}S_{i-1})^q 	\label{eq:1D-LR_deviate}
%\end{eqnarray}
%\begin{eqnarray}
%P_{i}\{R_i+\mathrm{d}R_i\}&=&P_{i+1}\{R_i+\mathrm{d}R_i\}		\nonumber\\
%P_{i}\{R_i+\mathrm{d}R_i\}&=&P_{i-1}\{R_i+\mathrm{d}R_i\}		\label{eq:1D-LR_deviate}
%\end{eqnarray}
By subtracting (\ref{eq:1D-LR}) from (\ref{eq:1D-LR_deviate}), we have:
%\begin{eqnarray}
%\mathrm{d}S_{i}(1-\beta\frac{qS_i^{q-1}}{2d_{i}}-\beta\frac{qS_i^{q-1}}{2d_{i-1}})=-\mathrm{d}P_{i}(\frac{1}{2d_{i}}+\frac{1}{2d_{i-1}})-\beta\frac{qS_{i+1}^{q-1}\mathrm{d}S_{i+1}}{2d_{i}}-\beta\frac{qS_{i-1}^{q-1}\mathrm{d}S_{i-1}}{2d_{i-1}} \nonumber
%\end{eqnarray}
%If we only focus on the situation when $q=1$, we have:
\begin{eqnarray} \label{eq:ds}
\mathrm{d}S_{i}(1-\frac{\beta}{2d_{i}}-\frac{\beta}{2d_{i-1}})&=&-\mathrm{d}P_{i}(\frac{1}{2d_{i}}+\frac{1}{2d_{i-1}})-\beta\frac{\mathrm{d}S_{i+1}}{2d_{i}} \nonumber\\
&&-\beta\frac{\mathrm{d}S_{i-1}}{2d_{i-1}}.
\end{eqnarray}
In fact, when we gradually change $P_i$  from $0$ to $P_{upper}$, $S_i$  will decrease. This is because with any small increment in $P_i$, at least customers at $i$'s boundary will  change to purchase at a different company. Due to area decrease, $i$ will lose even more customers. Similar analysis shows that $i+1$ and $i-1$ will have a positive market area deviation shown in (\ref{eq:ds}).

%\textbf{====I stop here I stop here I stop here I stop here July14 at 6pm====}

Hence, the right hand side of (\ref{eq:ds}) and $\mathrm{d}S_i$ in the left hand side will remain negative when $\mathrm{d}P_i>0$, implying that $1-\frac{\beta}{2d_{i}}-\frac{\beta}{2d_{i-1}}>0$, or equivalently $\beta<\frac{2d_{i}d_{i-1}}{d_i+d_{i-1}}$. Since during the price increment, $i$'s area will shrink and $i$' neighbors will change, once a close neighbor of $i$ with $\frac{2d_{i}d_{i-1}}{d_i+d_{i-1}}<\beta$ gains a positive market area (the new neighbor is denoted by $i-1$), $i$'s area will immediately shrink to $0$, i.e.,  ``wipe out" happens.

%\textbf{====the explanation below (till the end of the section) should be modified. It is not clear to me====}
Let us try to understand the meaning of upper bound  $\frac{2d_{i}d_{i-1}}{d_i+d_{i-1}}$ for $\beta$. %, we can see it in the other way. We are going to give the formula of $S_i$ and see when will $S_i\geq0$.
Suppose we hide company $i$ from the market just for the moment, i.e., $i$ is invisible to customers.  Since $i-1$ and $i+1$ remain activated,  they will share $i$'s previous market area.  Let $R'$ denote $i-1$ and $i+1$'s boundary, and $\tilde{P}_{i-1}\{x\}$ denote the aggregate price of $i-1$ at $x$ when $i$ is hidden. Since their price at boundary will be the same, i.e.,
 $\tilde{P}_{i-1}\{R'\}=\tilde{P}_{i+1}\{R'\}$, and note that $x_{i+1}-x_{i-1}=d_i+d_{i-1}$, we have:
%  \textbf{====State how one gets this====}
\begin{eqnarray}
R'= \frac{P_{i+1}-P_{i-1}-\beta S_{i+1}+\beta S_{i-1}}{2(d_i+d_{i-1})}+\frac{x_{i+1}+x_{i-1}}{2}\label{eq:R'}
\end{eqnarray}

%It is obvious that price at $R'$ is the most expensive in $[x_{i-1},x_{i+1}]$. Since $x_{i-1}<x_i<x_{i+1}$, if
Now let $i$ come back to the market.  If $i$'s price at $R'$ is higher than the price of $i-1$ or $i+1$ here, i.e., $\tilde{P}_{i+1}\{R'\}$, or $\tilde{P}_{i-1}\{R'\}$, then $i$ will have no market area.
This is because if $P_i\{R'\}>\tilde{P}_{i-1}\{R'\}$, for any $\delta>0$, at market point $R'-\delta$, we have $P_i\{R'-\delta\}-\tilde{P}_{i-1}\{R'-\delta\}=P_i\{R'\}-\tilde{P}_{i-1}\{R'\}+2d_{i-1}\delta>0$, similarly we always have $P_i\{R'+\delta\}>\tilde{P}_{i+1}\{R'+\delta\}$, which means $i$ will have no market area either at $x<R'$ or $x>R'$.

Consider the price gap between $i$ and $i-1$ at $R'$, by equation (\ref{eq:R'}), we have:
\begin{eqnarray}
\tilde{P}_{i-1}\{R'\}-P_i\{R'\}=\frac{2}{d_i+d_{i-1}}\Psi, 	\label{eq:beta-s}
\end{eqnarray}
where  $\Psi\triangleq\frac{P_{i+1}+d_i^2-\beta S_{i+1}-P_i}{2d_i}+\frac{P_{i-1}+d_{i-1}^2-\beta S_{i-1}-P_{i}}{2d_{i-1}}$.
If $\Psi<0$, which means $P_i\{R'\}>\tilde{P}_{i-1}\{R'\}$,
%then $i$ can not beat $i-1$ or $i+1$ even at their boundary,
accroding to the analysis above, $i$ will have no market area. Hence, if $i$ survives, we must have $\Psi>0$.
When $i$ is activated in the market, with the two equations in (\ref{eq:1D-LR}), we have
% let $\Psi=\frac{P_{i+1}+d_i^2-\beta S_i-P_i}{2d_i}+\frac{P_{i-1}+d_{i-1}^2-\beta S_{i-1}-P_{i}}{2d_{i-1}}$, we have:
\begin{eqnarray}
(1-\frac{\beta}{2d_{i}}-\frac{\beta}{2d_{i-1}})S_i&=&\frac{P_{i+1}+d_i^2-\beta S_{i+1}-P_i}{2d_i}\\
&&+\frac{P_{i-1}+d_{i-1}^2-\beta S_{i-1}-P_{i}}{2d_{i-1}} \nonumber \\
&=& \Psi 	\label{eq:beta-bound}
\end{eqnarray}

Interestingly,  $\Psi$ from (\ref{eq:beta-s}) appears in (\ref{eq:beta-bound}) again.
From (\ref{eq:beta-bound}), it is now clear that if $i$ survives, we must have $S_i>0$ and $\Psi>0$. Hence, we must have $(1-\frac{\beta}{2d_{i}}-\frac{\beta}{2d_{i-1}})>0$, i.e., $\beta<\frac{2d_{i}d_{i-1}}{d_i+d_{i-1}}$, if $i$ is not wiped out.

%\textbf{==Does this mean that $\beta<\frac{2d_{i}d_{i-1}}{d_i+d_{i-1}}$ is necessary for no wipe-out ``Yes''==}

%In (\ref{eq:beta-bound}) we can see that since $\Psi>0$, $\beta<\frac{2d_{i}d_{i-1}}{d_i+d_{i-1}}$ is exactly the condition when $S_i>0$, which is coincident to the results deduced from (\ref{eq:ds}).

\subsection{Existence of Equilibrium}
 So far we have introduced the ``wipe out'' phenomenon and explained  the bound  for $\beta$ in two ways. Now we show the existence of equilibrium.
\begin{theorem} \label{thm:1d-q1-exist}
Nash equilibrium always exists in the 1D market with $q=1$.
\end{theorem}
\begin{proof}(Sketch)(a) To prove this theorem, we first show that when there is no ``wipe out,'' i.e., $\beta$ is small such that $\beta<\frac{2d_{i}d_{i-1}}{d_i+d_{i-1}}, \forall i$, Nash equilibrium exists.
%This is similar to the case in Theorem \ref{thm:1d-q0-exist}, because
Under the no ``wipe out'' condition, one can guarantee that the utility function of a company is a piecewise continuous function.  The main difficulty in proving the quasi-concavity here is in showing that there exists a threshold price, such that the utility function is non-decreasing before the price exceeds it, and non-increasing after. %, especially at break points, i.e., junction points between two differential pieces.
%
%When ``wipe out'' cannot be completely eliminated,

(b) For the case when $\beta<\frac{2d_{i}d_{i-1}}{d_i+d_{i-1}}$ does not hold for all $i$, we construct an equilibrium with a set of carefully chosen ``activated'' companies, while  the others are considered ``hidden.'' Then, we show that this partial equilibrium among the chosen companies turns out to be exactly the equilibrium for all companies.
We can operate in the following way to give at least one equilibrium. We choose some companies to be hidden, which means they can be seen as not existed in the market, while the others being ``activated.''  For the simplicity in presentation below, we define a condition among those activated companies:
\begin{eqnarray}\beta<\frac{2d_{i}d_{i-1}}{d_i+d_{i-1}}\quad \forall\,\, \text{activated company } i .	\label{condition:activated-wipe-out}
\end{eqnarray}
  The way we choose is as follow:
(i) Hide all companies first. Then activate companies one by one until we can not activate any one more companies without bringing ``wipe out'' phenomenon into the market. That is, all activated companies satisfied the condition (\ref{condition:activated-wipe-out}), but activating any one of the hidden companies will violates it.
(ii) Among all activation schemes in (i), we can guarantee that there exists at least one scheme such that for any one hidden company, say $j$, it violates the inequality $\beta<\frac{2d_{j}d_{j-1}}{d_j+d_{j-1}}$. This is because if $\beta<\frac{2d_{j}d_{j-1}}{d_j+d_{j-1}}$ is satisfied for $j$, according to (i), there must exists one or two of its neighbors, say $k$, such that $\beta>\frac{2d_{k}d_{k-1}}{d_k+d_{k-1}}$, in that case, we can always activate $j$ and hide $k$.
(iii) Now that all activated companies satisfy the condition  (\ref{condition:activated-wipe-out}), while each  hidden company is being ``wiped out''.

Then for activated companies there exists equilibrium as proved in (a). Now let those hidden companies come back to the market with their prices being $P_{upper}$. From the analysis of equation (\ref{eq:beta-s}), we can see that hidden company $i$ will never survive by unilateral price change. For any activated company $j$, since hidden companies are being ''wiped out''(at the status of price being  $P_{upper}$ and have no market area), they can never survive by $j$'s unilateral action. Hence, $j$'s price will remain unchanged. This setting guarantees at least one equilibrium when $\beta<\frac{2d_{i}d_{i-1}}{d_i+d_{i-1}}, \forall i$ does not hold.
%
% in $\mathcal{N}$.
%because in this case, those ``hidden'' companies will find that they can never gain market area no matter what price they set.
%\textbf{====what is the meaning of this paragraph??====}
\end{proof}

\subsection{Necessary condition for equilibrium}
\begin{theorem}\label{thm:condition}
When a market with dimension $K$ is at a Nash equilibrium, for any company $i$, if it has no potential competitors,
\begin{equation} \label{eq:1d-part1}
P_i^*=c_i\cdot S_i^*,  \quad\quad \forall i \in \mathcal{N}_S
\end{equation}
where $c_i=-\frac{dP_{i}}{dS_{i}}|_{P_{i}=P_{i}^{*}}>0$.

Otherwise, we have:
\begin{equation} \label{ineq:equil}
c_i'\cdot S_{i}^*\leq P_{i}^*\leq c_i''\cdot S_{i}^*, \quad\quad \forall i \in \mathcal{N}_S.
\end{equation}
Here $c_i'=-\frac{dP_{i}}{dS_{i}}|_{P_{i}=P_{i}^{*}+}>0$, $c_i''=-\frac{dP_{i}}{dS_{i}}|_{P_{i}=P_{i}^{*}-}>0$.
Inequality (\ref{ineq:equil}) becomes an equality only when $i$'s potential competitor does not survive  at $P_i^*$.
\end{theorem}
In particular, when the market is 1D or 2D, and $\beta<<d_i$ holds for all $i$,   we have:
$c_i=-\frac{dP_{i}}{dS_{i}}|_{P_{i}=P_{i}^{*}}\approx\frac{1-\beta\sum_{j\in\mathcal{NR}(i)}\frac{l_{ij}}{2d_{ij}}}{\sum_{j\in\mathcal{NR}(i)}\frac{l_{ij}}{2d_{ij}}(\frac{\beta l_{ij}}{2d_{ij}}+1)}$.

Equation (\ref{eq:1d-part1}) is similar to Theorem \ref{thm:1d-q0-condition}, where the equilibrium price is proportional to market area.  Inequality (\ref{ineq:equil}) shows that potential competitors  restrict the pricing power of a company,  which results in lower selling price and hence lower profit compare to  (\ref{eq:1d-part1}). Since potential competitors are closer than actual competitors in market, once they survive, they
can cause significant harm to a company's market area and profit, due to the close similarity in  product features and high substitutability.
%\begin{remark} [Brand Factor] Theorem \ref{theorem:1d-ps-eq} provides the following interesting implications. (i) When there is no brand effect, i.e., brand factor is zero,  the equilibrium price of a company  is proportional to  its market area (market power) over the competition intensity  with its neighbors (boundary/distance).
%%This implies that lowering prices does not necessarily lead to market share increase, as observed in the smartphone market, where the intensity is high for customers preferring low-end phones.
%%
%(ii) When brand name has a positive effect in attracting customers, it is important to lower the price and seize more market area. Meanwhile,  companies with big market areas should try to  steer customers' shopping habits towards emphasizing on brand reputation (increase brand factor).
%(iii) Start-up companies should try to avoid markets or industry where the brand factor is large, in which case big companies are at strong advantage, and to avoid positioning at market points where competition is intense.
%\end{remark}

In addition, a company with larger market area can widen the market share gap by expanding the dimension of market, i.e., actively compete with weaker competitors using products different in multi-features. This is because market area gap in high dimension can be much bigger than the cases in 1D market, which in return will result in bigger aggregate price gap.
%\textbf{==theorem 8 is not well explained?? ==}

\section{Conclusion} \label{section:con}
In this paper, we study equilibrium properties based on a variant Hotelling model, considering brand name effect in the market by including a market area term into aggregate price.
%stylized yet general abstract market model.
We prove the existence of Nash equilibrium  in single-feature and dual-feature market,
%We provide the sufficient conditions for existence of  Nash equilibrium both in single-feature and dual-feature market, where company brand effect is modeled by a linear function of its market area times a market factor.
and also derive explicit characterizations of equilibrium prices and market areas. Our results reconcile the common belief that company's pricing power is proportional to its market area over competition intensity, and offer insight into pricing under brand name effect and market positioning.

%Discussion above give us a deeper understanding to the brand effect on the market. Depends on the $\beta$ value, we have provided  the following interesting implications.

Specifically, our results offer the following insight:
(i) When there is no brand effect or equivalent brand effect, i.e., $\beta$ or $q$ is zero,  the equilibrium price of a company  is proportional to  its market area (market power) over the competition intensity  with its neighbors (boundary over distance). This implies that lowering prices does not necessarily lead to market share increase, as observed in the smartphone market, where the competition intensity may be higher for customers who prefer low-end phones.
%This may be the explanation for the Samsung market share drop mentioned at the begining of the paper: Samsung substantially  lowered its average saling price(ASP) during the time period we discussed, which made the company to face a much intense competition with the low ASP competitors.
%
(ii) When brand name has a positive effect in attracting customers, it is important to lower the price and seize more market area. Meanwhile,  companies with big market areas should try to  steer customers' shopping habits towards emphasizing on brand reputation (increase brand factor).
(iii) New companies should try to avoid markets where the brand factor is large, in which case big companies are at strong advantage, and to avoid positioning at market points where competition is intense (many nearby companies), where due to the ``wipe out'' phenomenon, small companies may not be able to survive.

%In the end, we are going to give one open problems: does Nash equilibrium exist in 2D market when $q=1$?
%
%The situation in 2D market when $q=1$ remain unsolved, because the wipe out phenomenon appears much more often in higher dimension, which rapidly increase the difficulty of analysis.
%
%

\section{Appendix}
\subsection{Proof of Lemma \ref{lemma:one-survive}}
\begin{proof} Suppose for contradiction that for some $\mathbf{x_{0}}$, $S_{i}=0$ for all
$i\in\mathcal{O}'(\mathbf{x_{0}})$.
Suppose there is no company located on the point $\mathbf{x_{0}}$.
Define a neighborhood ball of \textbf{$\mathbf{x_{0}}$}, i.e.,
$Ball\{\mathbf{x_{0}}\}=\{\mathbf{x}|\thinspace||\mathbf{x}-\mathbf{x_{0}}||_{2}<\delta\}$. Let $\delta$ to be small enough such that there are no companies in this neighborhood ball. This is possible since the number of companies is finite.

The lowest and highest aggregate price that a company $i$ offers in this ball, denoted by $P_{i,low}$ and $P_{i,high}$, will be at the two intersect points between the ball and the line connecting $\mathbf{x_{i}}$ and $\mathbf{x_{0}}$.

Then, the price of any company $k\notin\mathcal{O}'(\mathbf{x}_{0})$
at $\mathbf{x}\in Ball\{\mathbf{x_0\}}$ satisfies  $P_{k}\{\mathbf{x}\}\geq P_{k,low}= P_{k}\{\mathbf{x_{0}}\}+\delta^{2}-2\delta\left\Vert\mathbf{x_k}-\mathbf{x_0}\right\Vert_2$,
and the price of any company $j\in\mathcal{O}'(\mathbf{x}_{0})$
satisfies  $P_{j}\{\mathbf{x}\}\leq P_{j,high}= P_{j}\{\mathbf{x_{0}}\}+\delta^{2}+2\delta\left\Vert\mathbf{x_j}-\mathbf{x_0}\right\Vert_2$.
If $\max_{j\in\mathcal{O}'(\mathbf{x_{0}})}P_{j,high}<\min_{k\notin\mathcal{O}'(\mathbf{x_{0}})}P_{k,low}$, we can make sure that all owners of market points in this ball will be in  $\mathcal{O}'(\mathbf{x}_{0})$.
Thus, we can
choose
$\delta<\frac{(\min_{k\notin\mathcal{O}'(\mathbf{x_{0}})}P_{k}\{\mathbf{x_{0}}\}-P_{j}\{\mathbf{x_{0}}\}) }{\max_{m\notin\mathcal{O}'(\mathbf{x_{0}}),n\in\mathcal{O}(\mathbf{x_{0}})}2(\left\Vert\mathbf{x_m}-\mathbf{x_0}\right\Vert_2+\left\Vert\mathbf{x_n}-\mathbf{x_0}\right\Vert_2)} $
to guarantee that all owners of market points in this ball are in  $\mathcal{O}'(\mathbf{x}_{0})$.
Hence, $\sum_{j\in\mathcal{O}'(\mathbf{x}_{0})}S_{j}>0$,
%Since they can not provide the same price at every point of the ball,
and $|\mathcal{O}'(\mathbf{x_{0}})|\leq N$,  which
violates the assumption that   $S_{j}=0,\forall j\in\mathcal{O}'(\mathbf{x}_{0})$.

As for the case when there is a company located exactly on  $\mathbf{x_{0}}$, if company  $\mathbf{x_{0}}\in \mathcal{O}'(\mathbf{x}_{0})$, the analysis will be the same. Otherwise the analysis still holds since we always have $P_j\{\mathbf{x_{0}}\}<P_k\{\mathbf{x_{0}}\}$ for $k\notin\mathcal{O}', j\in\mathcal{O}'$ regardless of $\delta$.
\end{proof}

%\subsection{Proof of Lemma \ref{lemma:one-point} }
%\begin{proof}Suppose the contrary  that $i$ provides the lowest price
%at more than one point, such as at $\mathbf{x_{1}},\mathbf{x_{2}}$,
%and it has $S_{i}=0$. From Lemma \ref{lemma:one-survive} we know that there exists
%at least one surviving owner $j$ of $\mathbf{x_{1}}$, with $S_{j}>0$
%.
%
%When market's dimension $K>1$, we can draw an axis with $\mathbf{x_{1}},\mathbf{x_{2}}$
%on it and project $i,j$ onto the axis to $x_{i},x_{j}$.
%Note that distance of companies $i,j$ to this axis will make no difference to the
%proof as they are constants,  therefore we can turn the high dimension case into an 1D problem.
%\textbf{why can one do this? shouldn't the other dimensions affect the overall distance?}
%
%
%Suppose $x_{2}-x_{1}=d>0$ and $x_{j}<x_{i}$ (the other situation
%will be similar).  Then, $P_{i}\{x_{1}\}=P_{j}\{x_{1}\}$,
%but $P_{i}\{x_{2}\}-P_{j}\{x_{2}\}=P_{i}\{x_{1}\}-P\{x_{1}\}+2d(x_{i}-x_{j})>0$, which leads to  an obvious contradiction.
%\end{proof}

\subsection{Proof of Theorem \ref{thm:convex}}
\begin{proof}
Consider the surviving companies.  We first see that the set
of market area where company $i$ and $j$ provide the same price is a
straight line(hyperplane in high dimension) given by:
\begin{eqnarray*}
\mathcal{L}_{i,j}&&=\{(x,y)|P_{i}+(x-x_{i})^{2}+(y-y_{i})^{2}-\beta S_{i}\\
&&\qquad-[P_{j}+(x-x_{j})^{2}+(y-y_{j})^{2}-\beta S_{j}]=0\}\\
&&=\{(x,y)|2x(x_{i}-x_{j})+2y(y_{i}-y_{j})-P_{i}+P_{j}\\
&&\qquad-\beta S_{j}+\beta S_{i}+x_{j}^{2}+y_{j}^{2}-x_{i}^{2}-x_{j}^{2}=0\}
\end{eqnarray*}
 This line divides the plane into 2 open half space $\mathcal{H}_{i}(i,j),\mathcal{H}_{j}(i,j)$,
where in $\mathcal{H}_{i}(i,j)$ we have $P_{i}\{(x,y)\}<P_{j}\{(x,y)\}$
and in $\mathcal{H}_{j}(i,j)$ we have $P_{j}\{(x,y)\}<P_{i}\{(x,y)\}$.
Thus, for any point $\mathbf{x}=(x,y)\in\mathcal{H}_{i}(i,j)$, we have
$j\notin\mathcal{O}(\mathbf{x})$ and vice versa. Let $\mathcal{U}_{i}=\cap_{j\in\mathcal{N},j\neq i}\mathcal{H}_{i}(i,j)$,
since $\mathcal{U}_{i}$ is the intersection of half
spaces, $\mathcal{U}_{i}$ is a convex polygon.

Now we prove that $\mathcal{U}_{i}=\mathcal{S}_{i}$. For
any $\mathbf{x}\in\mathcal{S}_{i}$, we have $\mathcal{O}(\mathbf{x})=\{i\}$.
Hence,  $\mathbf{x}\in\mathcal{H}_{i}(i,j),\forall j\neq i$. Thus, $\mathbf{x}\in\cap_{j\in\mathcal{N},j\neq i}\mathcal{H}_{i}(i,j)=\mathcal{U}_{i}$ and
$\mathcal{S}_{i}\in\mathcal{U}_{i}$. For any $\mathbf{x}\in\mathcal{U}_{i}$,
we see that $j\notin\mathcal{O}(\mathbf{x})$ if $j\neq i$. Since $\mathcal{O}(\mathbf{x})\neq\emptyset$, $\mathcal{O}(\mathbf{x})=\{i\}$, $\mathcal{U}_{i}\in\mathcal{S}_{i}$ and $\mathcal{U}_{i}=\mathcal{S}_{i}$.

It remains to show that $\mathcal{S}'_{i}=\mathbf{cl}(\mathcal{U}_{i})$.
For any $\mathbf{x}\in\mathbf{cl}(\mathcal{U}_{i})$, we have $P_{i}\{\mathbf{x}\}\leq P_{j}\{\mathbf{x}\},\forall j\neq i$. Hence,  $i\in\mathcal{O}(\mathbf{x})$ and $\mathbf{x}\in\mathcal{S}'_{i}$, which implies
$\mathbf{cl}(\mathcal{U}_{i})\in\mathcal{S}'_{i}$. Now consider any $\mathbf{x}\in\mathcal{S}'_{i}$.
Since for all $\mathbf{x}'\in\mathbf{cl}(\mathcal{U}_{i}^{C})$,
where $\mathbf{cl}(\mathcal{U}_{i}^{C})=\mathcal{M}-\mathbf{cl}(\mathcal{U}_{i})$
is the complementary set of $\mathbf{cl}(\mathcal{U}_{i})$, we have
$i\notin\mathcal{O}(\mathbf{x}')$. Thus, we must have $\mathbf{x}\in\mathbf{cl}(\mathcal{U}_{i})$ and $\mathcal{S}'_{i}\in\mathbf{cl}(\mathcal{U}_{i})$.

In conclusion, $\mathcal{S}'_{i}=\mathbf{cl}(\mathcal{U}_{i})$,
and $\mathcal{S}'_{i}=\mathbf{cl}(\mathcal{S}{}_{i})$.
\end{proof}

\subsection{Proof of Theorem \ref{thm:1d-q0-exist}}
Before proving the theorem, we first have the following lemma, whose proof will be given after the proof of Theorem \ref{thm:1d-q0-exist}.
\begin{lemma}\label{lemma:neighbor}
In the 1D market,  let $x_i$ denote the i-th surviving company in increasing order of their coordinates.  We have  $L_{i} = R_{i-1}$ for all $i \in \mathcal{N}_S(\mathbf{P})$.
%the left neighbor of $i$ is $i-1$ andthe right neighbor of $i$ is $i+1$.
\end{lemma}

\begin{proof} (Theorem \ref{thm:1d-q0-exist})
From
%(\ref{eq:p-q0})
(7) we know that when $q=0$, the brand effect factor will be the same to every companies, regardless of the $\beta$ value. Thus, we consider $\beta=0$.
%the result for $\beta=0$ applies to general $\beta\geq0$.

%Let $\beta=0$.
For any companies $i$, other companies' aggregate prices at $\mathbf{x_i}=x_i$ are always positive due to distance.
This means that for any companies $i$, given $P_{-i}$, $i$ can always set a price $P_i$ such that $S_i>0$. Therefore, every company will survive in the market.

For company $i$, by Theorem \ref{thm:convex} and Lemma \ref{lemma:neighbor}, we know that $i$'s market area will be an interval on the line, with neighbors  always being $i-1, i+1$, which satisfy $x_{i-1}<x_i<x_{i+1}$. The right boundary of $i-1$ and $i$ can be calculated as:
\begin{eqnarray*}
P_{i}+(R_{i-1}-x_{i})^{2}&=&P_{i-1}+(R_{i-1}-x_{i-1})^{2} \\
P_{i}+(R_{i}-x_{i})^{2}&=&P_{i+1}+(R_{i}-x_{i+1})^{2}
\end{eqnarray*}
Thus,   the  utility  $W_i$ is given by:
\begin{eqnarray*}
W_i &=& P_i S_i = P_i(R_i-R_{i-1})\\
&=&-P_{i}^{2}(\frac{1}{2d_{i}}+\frac{1}{2d_{i-1}}) \nonumber\\
&&+P_{i}(\frac{P_{i+1}+x_{i+1}^{2}-x_{i}^{2}}{2d_{i}}-\frac{P_{i-1}-x_{i}^{2}+x_{i-1}^{2}}{2d_{i-1}}).
\end{eqnarray*}
Hence, the profit will be a downward parabola for $P_i\in[0,\frac{\frac{P_{i+1}+x_{i+1}^{2}-x_{i}^{2}}{2d_{i}}-\frac{P_{i-1}-x_{i}^{2}+x_{i-1}^{2}}{2d_{i-1}}}{\frac{1}{2d_{i}}+\frac{1}{2d_{i-1}}}]$ and zero otherwise.

We see then the payoff of $i$ is continuous in $\mathbf{P}$, quasiconcave in $P_i$. Hence from theorem  \ref {thm:1952}, we see that Nash equilibrium always exists.
\end{proof}

\subsection{Proof of Lemma \ref{lemma:neighbor}}
\begin{proof}
We have seen that $S_{i}=R_{i}-L_{i}>0$ for each surviving company. Consider the right
boundary of company $i$. We  prove $R_i=L_{i+1}$ by contradiction.  Suppose instead $R_i=L_j$ for some $j\neq i+1$. Then,
%$j\neq i+1$
%is the right neighbor of $i$, so
\begin{equation}
P_{i}\{R_{i}\}=P_{j}\{R_{i}\}<P_{k}\{R_{i}\},\,\,\,\forall k\neq i,j
\end{equation}
We have the following cases:
%Each possible situation is discussed below.
\begin{enumerate}
\item Suppose $j<i$. Then,
%We can show that the right boundary of $i$ can not happen between
%$i$ and $j<i$. Otherwise if $j<i$, we have:
\begin{equation}
P_{i}\{R_{i}\}=P_{j}\{R_{i}\}
\end{equation}
By choosing some small positive $\delta$ such that $x'=R_{i}-\delta\in\mathcal{S}_{i}$,
 we have:
\begin{equation}
P_{i}\{x'\}-P_{j}\{x'\}=2\delta(x_{i}-x_{j})>0
\end{equation}
which is an obvious contradiction to $x'\in\mathcal{S}_{i}$.
\item Now suppose $j>i+1$, we show below that $i+1$ will have no market area.

\begin{enumerate}
\item Consider the interval $[R_{i},\infty)$. Since $R_i=L_j$, we know that $P_{j}\{R_{i}\}\leq P_{i+1}\{R_{i}\}$. Let $x'=R_{i}+\Delta x$ for some small $\Delta x$.  We have:
\begin{align*}
&P_{j}\{x'\}\\
=&P_{j}+(R_{i}-x_{j})^{2}+2\Delta x(R_{i}-x_{j})+(\Delta x)^{2}-\beta S_{j}\\
=&P_{j}\{R_{i}\}+2\Delta x(R_{i}-x_{j})+(\Delta x)^{2}
\end{align*}
Similarly,
\[
P_{i+1}\{x'\}=P_{i+1}\{R_{i}\}+2\Delta x(R_{i}-x_{i+1})+(\Delta x)^{2}.
\]
Thus,
\begin{align*}
&P_{i+1}\{x'\}-P_{j}\{x'\}\geq P_{i+1}\{R_{i}\}-P_{j}\{R_{i}\}\\
&+2\Delta x(x_{j}-x_{i+1})>0,\forall x'>R_{i}
\end{align*}
This means that there is no $x\in[R_i, \infty)$ such that $\mathcal{O}(x)={i+1}$, i.e., company $i+1$ has zero market area.

\item Consider the interval $(-\infty,R_{i})$. We know that $P_{i}\{R_{i}\}\leq P_{i+1}\{R_{i}\}$.
Let $x'=R_{i}-\Delta x$.  Similar to the
argument above, we obtain:
\begin{align}
&P_{i+1}\{x'\}-P_{i}\{x'\} \nonumber\\
=&P_{i+1}\{R_{i}\}-P_{i}\{R_{i}\}+2\Delta x(x_{i+1}-x_{i})>0,\forall x'<R_{i}
\end{align}

\end{enumerate}
\end{enumerate}
From the above, we see that the right neighbor of $i$ must be $i+1$. Otherwise $i+1$ will not survive.
\end{proof}

\subsection{Proof of Theorem \ref{thm:1d-q0-condition}}
\begin{proof}
When the market is at equilibrium,  we have:
\begin{equation}\nonumber
\frac{\mathrm{d}W_{i}}{\mathrm{d}P_{i}}=S_{i}+P_{i}\frac{\mathrm{d}S_{i}}{\mathrm{d}P_{i}}=0
\end{equation}
and
\begin{equation}\nonumber
\frac{\mathrm{d}S_{i}}{\mathrm{d}P_{i}}=\frac{1}{2d_{i}}+\frac{1}{2d_{i-1}}.
\end{equation}
Combining these two equations, we have proved the theorem:
\begin{equation}
\frac{P_i}{S_i}=\frac{2d_{i}d_{i-1}}{d_i+d_{i-1}}, \forall i\in \mathcal{N}
\end{equation}
\end{proof}

%%%%%%%%%%%     2D existence proof:

\subsection{Proof of Theorem \ref{theorem:2d-exist}}\label{section:proof-2dtheroem}
\begin{proof}
Similar to Theorem \ref{thm:1d-q0-exist}, we consider $\beta=0$.
First we  show that $S_{i}(P_{i})$ and $W_{i}=P_{i}S_{i}$ are continuous.
 For one neighbor $j$ of $i$, set an axis with
$\mathbf{x_{i}},\mathbf{x_{j}}$ on it. Then, the boundary of $i$ and $j$ will be
a straight line perpendicular to this axis. Let $x_{i}',x_{j}'$
denote the coordinate of $i,j$ on the axis and choose the direction
of the axis so that $x_{i}'<x_{j}'$ (see Fig. \ref{fig:2D-1D}). Denote $d_{ij}=x_{j}-x_{i}$.
Doing so, we turn the problem between $i,j$ into an 1D problem. Let
$x_{ij}'$ denote the coordinate of the boundary on this axis. We
have the following equations when $P_{i}$ is changed by  a small  $\mathrm{d}P_{i}$, where
 $\mathrm{d}x_{ij}'$ denotes the corresponding deviation of $x_{ij}'$:
\begin{eqnarray}
P_{i}+(x_{i}'-x_{ij}')^{2}&=&P_{j}+(x_{j}'-x_{ij}')^{2}\\
P_{i}+\mathrm{d}P_{i}+(x_{i}'-x_{ij}'-dx')^{2}&=&P_{j}+(x_{j}'-x_{ij}'-\mathrm{d}x')^{2}.
\end{eqnarray}
Thus, we have:
\begin{equation}
\mathrm{d}x_{ij}=-\frac{\mathrm{d}P_i}{2d_{ij}}.
\end{equation}
Hence,
\begin{align}
\lim_{dP_{i}\rightarrow0}S_i(P_{i}+\mathrm{d}P_{i})-S_i(P_{i})=&\sum_{j\in\mathcal{NB}(i)}l_{ij}\mathrm{d}x_{i,j} \nonumber\\ =&\sum_{j\in\mathcal{NB}(i)}-\frac{l_{ij}\mathrm{d}P_{i}}{2d_{ij}}=0,
\end{align}
where $l_{ij}$ is the length of boundary between $i,j$. Thus, $S_{i}(P_{i})$ is continuous and piecewise differentiable, and
\begin{equation}
\frac{\mathrm{d}S_{i}}{\mathrm{d}P_{i}}=\sum_{j\in\mathcal{NB}(i)}-\frac{l_{ij}}{2d_{ij}}<0.
\end{equation}
Since $W_{i}(P_{i})=P_{i}S_{i}$,  $W_{i}$ is also continuous in $\mathbf{P}$.  %We know that  $W_{i}(0)=0$ and that at some price $P_{max}$, $i$'s market area will drop to zero.
Now let $P_{i}$ increase
from $0$ to $P_{upper}$ (fixing $P_{-i}$). The polygon $\mathcal{S}_{i}'$ will continuously
shrink and eventually disappear.
If there exists a positive price such that profit becomes zero, denote it as $P_{i,max}$. We know that for $P_i\in [P_{i,max}, P_{upper}]$, $W_i(P_i)=0$.

Consider any time during the shrinking process and consider an edge of the polygon, say
the boundary between $i$ and $j$. We know that there must be at least one more company besides $i,j$ that provides the same lowest price at each of these two  endpoints (See Fig. \ref{fig:dldp}(a)).
%%%%%%%%%%%%%%% figure 1
\begin{figure}
%\centerline{\includegraphics[scale=0.4]{dldp-omni3}}
\centerline{\includegraphics[scale=0.28]{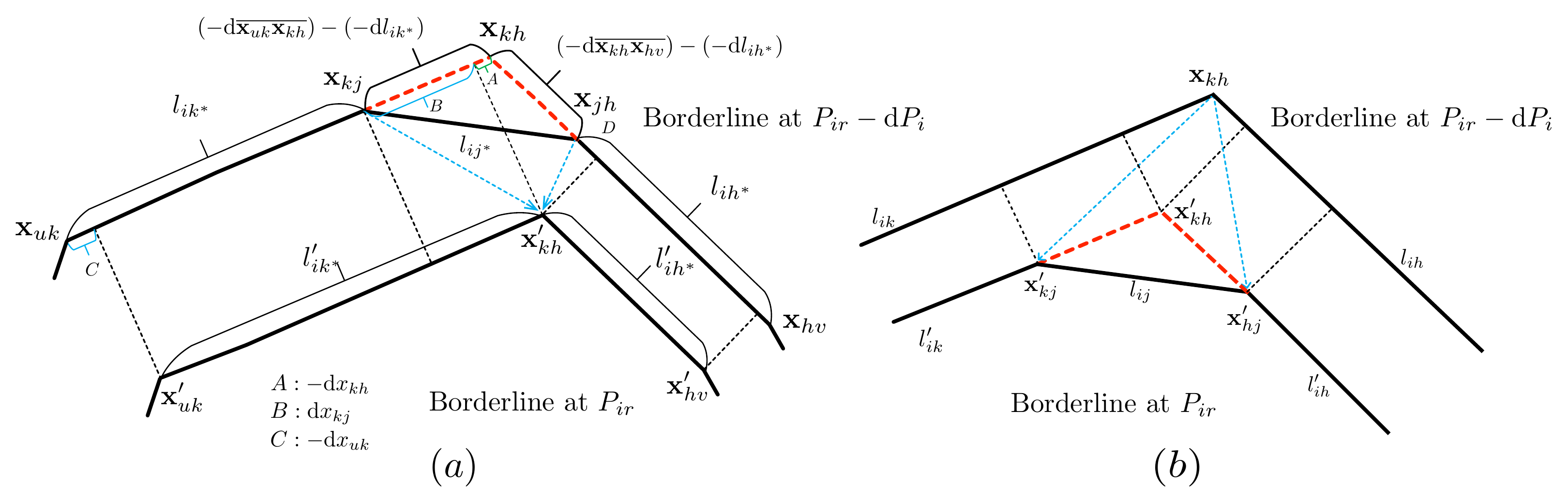}}
\caption{Boundary lines at prices $P_{i}$ and $P_{i}-\mathrm{d}P_{i} $ are shown in the figure. Figure (a) captures the moment before $l_{ij}$ disappears, which is thoroughly discussed in the proof. Here $\mathbf{x}_{kh}$ donotes the intersection point of the extended boundary line of $l_{ik}$ and $l_{ih}$, others are similar.  From this figure we see that at $P_{i}-\mathrm{d}P_{i}$, the boundary line is $\overline{\mathbf{x}_{uk}\mathbf{x}_{kj}\mathbf{x}_{jh}\mathbf{x}_{hv}}$, which becomes $\overline{\mathbf{x'}_{uk}\mathbf{x}'_{kh}\mathbf{x}'_{hv}}$ at $P_{i}$.
Coefficients $M_{ik},N_{ik},C_{ik,1},C_{ik,2}$ $M_{ih},N_{ih},C_{ih,1},C_{ih,2}$ have changed along with the shape of $i$'s market area. Discontinuity of $ \frac{\mathrm{d}^{2}W_{i}}{\mathrm{d}P_{i}^{2}}$ arises from this saltus. Figure (b) captures the situation when a potential competitor $j$ starts to gain profit. The proof is similar to (a) using the triangle inequality.}
\label{fig:dldp}
\end{figure}
Denote them by $k$ and $h$. Then, we can use  $\mathbf{x_{kj}}=(x_{kj},y_{kj})$,
$\mathbf{x_{jh}}=(x_{jh},y_{jh})$ to denote two endpoint of this boundary
line, i.e.,
  $P_{i}\{\mathbf{x_{kj}}\}=P_{j}\{\mathbf{x_{kj}}\}=P_{k}\{\mathbf{x_{kj}}\}$,
$P_{i}\{\mathbf{x_{jh}}\}=P_{j}\{\mathbf{x_{jh}}\}=P_{h}\{\mathbf{x_{jh}}\}$, and
we have following equations:
\begin{align*}
P_{i}+(x_{kj}-x_{i})^{2}+(y_{kj}-y_{i})^{2}&=P_{j}+(x_{kj}-x_{j})^{2}+(y_{kj}-y_{j})^{2}\\
P_{k}+(x_{kj}-x_{k})^{2}+(y_{kj}-y_{k})^{2}&=P_{j}+(x_{kj}-x_{j})^{2}+(y_{kj}-y_{j})^{2}\\
P_{i}+(x_{jh}-x_{i})^{2}+(y_{jh}-y_{i})^{2}&=P_{j}+(x_{jh}-x_{j})^{2}+(y_{jh}-y_{j})^{2}\\
P_{h}+(x_{jh}-x_{h})^{2}+(y_{jh}-y_{h})^{2}&=P_{j}+(x_{jh}-x_{j})^{2}+(y_{jh}-y_{j})^{2}
\end{align*}

In the following, we first express $l_{ij}$ as an explicit function of $P_i$ with these equations. Then, we calculate the second derivative of $W_{i}$ and  prove the quasi-concavity of $W_{i}$.
To this end, define
\begin{eqnarray*}
T_{ji}&=&P_{j}+x_{j}^{2}+y_{j}^{2}-x_{i}^{2}-y_{i}^{2}\\
T_{jk}&=&P_{j}-P_{k}+x_{j}^{2}+y_{j}^{2}-x_{k}^{2}-y_{k}^{2}\\
T_{jh}&=&P_{j}-P_{h}+x_{j}^{2}+y_{j}^{2}-x_{h}^{2}-y_{h}^{2},
\end{eqnarray*}
and let
\begin{eqnarray*}
|A_{kj}|&=&\left|\begin{array}{cc}
2(x_{j}-x_{i}) & 2(y_{j}-y_{i})\\
2(x_{j}-x_{k}) & 2(y_{j}-y_{k})
\end{array}\right|\\
|A_{jh}|&=&\left|\begin{array}{cc}
2(x_{j}-x_{i}) & 2(y_{j}-y_{i})\\
2(x_{j}-x_{h}) & 2(y_{j}-y_{h})
\end{array}\right|
\end{eqnarray*}
The endpoints $\mathbf{x_{m}}$ and $\mathbf{x_{n}}$ can then be expressed as:
\begin{align*}
x_{kj}&=\frac{\left|\begin{array}{cc}
T_{ji}-P_{i} & 2(y_{j}-y_{i})\\
T_{jk} & 2(y_{j}-y_{k})
\end{array}\right|}{|A_{kj}|}, \\
 y_{kj}&=\frac{\left|\begin{array}{cc}
2(x_{j}-x_{i}) & T_{ji}-P_{i}\\
2(x_{j}-x_{k}) & T_{jk}
\end{array}\right|}{|A_{kj}|},\\
x_{jh}&=\frac{\left|\begin{array}{cc}
T_{ji}-P_{i} & 2(y_{j}-y_{i})\\
T_{jh} & 2(y_{j}-y_{h})
\end{array}\right|}{|A_{jh}|},\\
 y_{jh}&=\frac{\left|\begin{array}{cc}
2(x_{j}-x_{i}) & T_{ji}-P_{i}\\
2(x_{j}-x_{h}) & T_{jh}
\end{array}\right|}{|A_{jh}|}.
\end{align*}
Thus, the length of the boundary line $l_{ij}$ can be expressed as:
\begin{eqnarray}
l_{ij}^{2}&&=(x_{kj}-x_{jh})^{2}+(y_{kj}-y_{jh})^{2} \nonumber \\
&&=(\frac{\left|\begin{array}{cc}
T_{ji} & 2(y_{j}-y_{i})\\
T_{kj} & 2(y_{j}-y_{k})
\end{array}\right|-2P_{i}(y_{j}-y_{k})}{|A_{kj}|}\nonumber\\
&&\quad -\frac{\left|\begin{array}{cc}
T_{ji} & 2(y_{j}-y_{i})\\
T_{jh} & 2(y_{j}-y_{h})
\end{array}\right|-2P_{i}(y_{j}-y_{h})}{|A_{jh}|})^{2}  \nonumber  \\
&&\quad+(\frac{\left|\begin{array}{cc}
2(x_{j}-x_{i}) & T_{ji}\\
2(x_{j}-x_{k}) & T_{kj}
\end{array}\right|+2P_{i}(x_{j}-x_{k})}{|A_{kj}|} \nonumber\\
&&\quad -\frac{\left|\begin{array}{cc}
2(x_{j}-x_{i}) & T_{ji}\\
2(x_{j}-x_{h}) & T_{jh}
\end{array}\right|+2P_{i}(x_{j}-x_{h})}{|A_{jh}|})^{2}  \nonumber  \\
&&=(2M_{ij}P_{i}+C_{ij,1})^{2}+(2N_{ij}P_{i}+C_{ij,2})^{2}.	 \label{eq:l_ij-square}
\end{eqnarray}
In (\ref{eq:l_ij-square}), we have used
\begin{align*}
M_{ij}&=\frac{y_{j}-y_{h}}{|A_{jh}|}-\frac{y_{j}-y_{k}}{|A_{kj}|}\\
N_{ij}&=\frac{x_{j}-x_{k}}{|A_{kj}|}-\frac{x_{j}-x_{h}}{|A_{jh}|}\\
C_{ij,1}&=\frac{\left|\begin{array}{cc}
T_{ji} & 2(y_{j}-y_{i})\\
T_{kj} & 2(y_{j}-y_{k})
\end{array}\right|}{|A_{kj}|}-\frac{\left|\begin{array}{cc}
T_{ji} & 2(y_{j}-y_{i})\\
T_{jh} & 2(y_{j}-y_{h})
\end{array}\right|}{|A_{jh}|}\\
C_{ij,2}&=\frac{\left|\begin{array}{cc}
2(x_{j}-x_{i}) & T_{ji}\\
2(x_{j}-x_{k}) & T_{kj}
\end{array}\right|}{|A_{kj}|}-\frac{\left|\begin{array}{cc}
2(x_{j}-x_{i}) & T_{ji}\\
2(x_{j}-x_{h}) & T_{jh}
\end{array}\right|}{|A_{jh}|}
\end{align*}.
%\[
%\begin{array}{cc}
%M_{ij}=\frac{y_{j}-y_{h}}{|A_{jh}|}-\frac{y_{j}-y_{k}}{|A_{kj}|}, & C_{ij,1}=\frac{\left|\begin{array}{cc}
%T_{ji} & 2(y_{j}-y_{i})\\
%T_{kj} & 2(y_{j}-y_{k})
%\end{array}\right|}{|A_{kj}|}-\frac{\left|\begin{array}{cc}
%T_{ji} & 2(y_{j}-y_{i})\\
%T_{jh} & 2(y_{j}-y_{h})
%\end{array}\right|}{|A_{jh}|}\\
%N_{ij}=\frac{x_{j}-x_{k}}{|A_{kj}|}-\frac{x_{j}-x_{h}}{|A_{jh}|}, & C_{ij,2}=\frac{\left|\begin{array}{cc}
%2(x_{j}-x_{i}) & T_{ji}\\
%2(x_{j}-x_{k}) & T_{kj}
%\end{array}\right|}{|A_{kj}|}-\frac{\left|\begin{array}{cc}
%2(x_{j}-x_{i}) & T_{ji}\\
%2(x_{j}-x_{h}) & T_{jh}
%\end{array}\right|}{|A_{jh}|}
%\end{array}.
%\]
Therefore, $l_{ij}$ is continuous and differentiable, and
\begin{equation}
\frac{\mathrm{d}l_{ij}}{\mathrm{d}P_{i}}=\frac{1}{2l_{ij}}\cdot\frac{\mathrm{d}(l_{ij})^{2}}{\mathrm{d}P_{i}}=\frac{[4(M_{ij}^{2}+N_{ij}^{2})P_{i}+2(M_{ij}C_{ij,1}+N_{ij}C_{ij,2})]}{l_{ij}}
\end{equation}
So far, we have looked at  $l_{ij}$. The next step is to prove $W_{i}$'s quasi-concavity. We have:
\begin{equation}\label{eq:dwdp}
\frac{\mathrm{d}W_{i}}{\mathrm{d}P_{i}}=S_{i}+P_{i}\frac{\mathrm{d}S_{i}}{\mathrm{d}P_{i}}=S_{i}-P_{i}\sum_{j\in\mathcal{NB}(i)}\frac{l_{ij}}{2d_{ij}}
\end{equation}
 $\frac{dW_{i}}{dP_{i}}$ is also piecewise differentiable.
Moreover,
\begin{equation}
\frac{\mathrm{d}^{2}W_{i}}{\mathrm{d}P_{i}^{2}}=-2\sum_{j\in\mathcal{NB}(i)}\frac{l_{ij}}{2d_{ij}}-P_{i}\sum_{j\in\mathcal{NB}(i)}\frac{1}{2d_{ij}}\cdot\frac{\mathrm{d}l_{ij}}{\mathrm{d}P_{i}}\label{eq:dw2=00003Dsum}
\end{equation}
Note that if for any $P_{i0}$ that satisfies $\frac{\mathrm{d}W_{i}}{\mathrm{d}P_{i}}|_{P_{i}=P_{i0}}=0$, we always have:
\begin{equation}
\frac{\mathrm{d^{2}}W_{i}}{\mathrm{d}P_{i}^{2}}|_{P_{i}=P_{i0}}\leq0,
\end{equation}
then $W_{i}$ will be quasiconcave.
The most challenging part here is when during the shrinking process,
the shape of polygon $\mathcal{S}_{i}'$ may change.

For example, during the increment of $P_{i}$, $l_{ij}$  can shrink to a point and eventually disappear, at some $P_{i}$ when $j$ is no longer neighbor of $i$.
Therefore, $\frac{\mathrm{d}^{2}W_{i}}{\mathrm{d}P_{i}^{2}}$ will be a piecewise
function, with each piece starting at the time when  one boundary line of
$i$ disappears. Yet from (\ref{eq:dwdp}) we know that  $\frac{\mathrm{d}W_{i}}{\mathrm{d}P_{i}}$ is continuous because the moment $l_{ik}$ starts to appear or disappear, we have $l_{ik}=0$.
 Fig. \ref{fig:dldp}(a) shows the boundary line of $i$ at price $P_{ir}-\mathrm{d}P$ and $P_{ir}$, where at $P_{ir}$ boundary $l_{ij^*}$ disappears. In Fig. \ref{fig:dldp}(a), $l_{ij^*}$ intersects with $l_{ik^*}$ and $l_{ih^*}$ before disappearing.  The two endpoints of $l_{ij^*}$ here are denoted by $\mathbf{x}_{kj}$ and  $\mathbf{x}_{hj}$.
Situation that a potential competitor $j$ starts to gain profit is presented in Fig. \ref{fig:dldp}(b). The   proof of this situation is similar and is hence omitted for brevity.

Now for any piece of $W_{i}$, we can write it as:
\begin{align}
&\frac{\mathrm{d}^{2}W_{i}}{\mathrm{d}P_{i}^{2}}\nonumber\\
=&-\sum_{j\in\mathcal{NB}(i)}\frac{l_{ij}^{2}+P_{i}[2(M_{ij}^{2}+N_{ij}^{2})P_{i}+(M_{ij}C_{ij,1}+N_{ij}C_{ij,2})]}{d_{ij}l_{ij}}
\end{align}
Plugging (\ref{eq:l_ij-square}) into the above, we get:
\begin{align} \label{eq:d2w}
&\frac{\mathrm{d}^{2}W_{i}}{\mathrm{d}P_{i}^{2}}
=-\sum_{j\in\mathcal{NB}(i)}(\frac{6(M_{ij}^{2}+N_{ij}^{2})P_{i}^{2}+5(M_{ij}C_{ij,1}+N_{ij}C_{ij,2})P_{i}}{d_{ij}l_{ij}}\nonumber\\
&+\frac{C_{ij,1}^{2}+C_{ij,2}^{2}}{d_{ij}l_{ij}})	
=-\sum_{j\in\mathcal{NB}(i)}\frac{6(M_{ij}^{2}+N_{ij}^{2})}{d_{ij}l_{ij}}\cdot P_{i}^{2} \nonumber \\
&-\sum_{j\in\mathcal{NB}(i)}\frac{5(M_{ij}C_{ij,1}+N_{ij}C_{ij,2})}{d_{ij}l_{ij}}\cdot P_{i}-\sum_{j\in\mathcal{NB}(i)}\frac{C_{ij,1}^{2}+C_{ij,2}^{2}}{d_{ij}l_{ij}}.
%\label{eq:parabola of W}.
\end{align}
Thus, every piece of $W_{i}$ is  the sum of a series of parabolas, which is a downward parabola. %,
In each piece of $W_{i}$, we can calculate the price $P_{ij}^{*}$, at which $l_{ij}$ shrinks to zero, i.e.,
\begin{equation}
l_{ij}^{2}=(2M_{ij}P_{ij}^{*}+C_{ij,1})^{2}+(2N_{ij}P_{ij}^{*}+C_{ij,2})^{2}=0,\thinspace\thinspace\thinspace\thinspace\thinspace\forall j\in\mathcal{NB}(i)\label{eq:l-square*}
\end{equation}
Thus, the right discontinuity point $P_{ir}$ of each piece of $\frac{\mathrm{d}^{2}W_{i}}{\mathrm{d}P_{i}^{2}}$
is decided by the time when $P_{i}$ reach the smallest $P_{ij}^{*}$:

\begin{equation}
P_{ir}=\min_{j}P_{ij}^{*}\thinspace\thinspace\thinspace\thinspace\thinspace j\in\mathcal{NB}(i). \label{eq:p-ir}
\end{equation}

This means at $P_{i}=P_{ir}$, there exists only one $j^{*}\in\mathcal{NB}(i)$,
such that length of boundary line $l_{ij^*}$ becomes $0$.
Starting from $P_{i}=P_{ij}^{*}$, $\mathcal{S}_{i}'$ change its shape
and the coefficients related to $j^{*}$, i.e., $M_{ik},N_{ik},C_{ik,1},C_{ik,2}$,$M_{ih},N_{ih},C_{ih,1},C_{ih,2}$,  also change.

%\begin{figure}
%\centerline{\includegraphics[scale=0.2]{dw2dp-v1}}
%\caption{An example of  $\frac{\mathrm{d}^{2}W_{i}}{\mathrm{d}P_{i}^{2}}$. Each piece of  $\frac{\mathrm{d}^{2}W_{i}}{\mathrm{d}P_{i}^{2}}$ is a fragment from an increasing downward parabola, therefore  $\frac{\mathrm{d}^{2}W_{i}}{\mathrm{d}P_{i}^{2}}$ is a monotonically increasing function on $[0,P_{upper}]$.}
%\label{fig:dw2dp}
%\end{figure}

Let $f(P_{i};i,j)=\frac{6(M_{ij}^{2}+N_{ij}^{2})P_{i}^{2}+5(M_{ij}C_{ij,1}+N_{ij}C_{ij,2})P_{i}+C_{ij,1}^{2}+C_{ij,2}^{2}}{d_{ij}l_{ij}}$, which is a parabola function of $P_{i}$. From (\ref{eq:d2w}), we know that $\frac{\mathrm{d}^{2}W_{i}}{\mathrm{d}P_{i}^{2}}=-\sum_{j\in\mathcal{NB}(i)}f(P_{i};i,j)$, we can take $\frac{\mathrm{d}^{2}W_{i}}{\mathrm{d}P_{i}^{2}}$
as sum of a group of parabolas $f(P_{i};i,j)$.
In any pieces of $\frac{\mathrm{d}^{2}W_{i}}{\mathrm{d}P_{i}^{2}}$, we know from equations  (\ref{eq:l-square*}) and (\ref{eq:p-ir})
that $P_{ir}=-\frac{C_{ij^*,1}}{2M_{ij^*}}=-\frac{C_{ij^*,2}}{2N_{ij^*}}$
and $P_{ij}^{*}=-\frac{C_{ij^{*},1}}{2M_{ij^{*}}}=-\frac{C_{ij^{*},2}}{2N_{ij^{*}}}>P_{ir}$ for any
$j\in\mathcal{NB}(i)$ and $j\neq j^{*}$. Since all the axis of
symmetry of $f(P_{i};i,j)$ are
larger than $P_{ir}$, the axis of symmetry of $\frac{\mathrm{d}^{2}W_{i}}{\mathrm{d}P_{i}^{2}}$
will also larger than $P_{ir,}$. Hence, each piece of $\frac{\mathrm{d}^{2}W_{i}}{\mathrm{d}P_{i}^{2}}$
is a piece of downward parabola on the left of its symmetry axis and
is increasing (See Fig. \ref{fig:dw2dp} for an example). Let $\frac{\mathrm{d}^{2}W_{i}'}{\mathrm{d}P_{i}^{2}}|_{P_{i}=P_{ir}}$
denote the value of the next piece of parabola $\frac{\mathrm{d}^{2}W_{i}'}{\mathrm{d}P_{i}^{2}}$
starting at $P_{ir}$. If we can prove $\frac{\mathrm{d}^{2}W_{i}'}{\mathrm{d}P_{i}^{2}}|_{P_{i}=P_{ir}}>\frac{\mathrm{d}^{2}W_{i}}{\mathrm{d}P_{i}^{2}}|_{P_{i}=P_{ir}}$,
then $\frac{\mathrm{d}^{2}W_{i}}{\mathrm{d}P_{i}^{2}}$
is an increasing function on $[0,P_{max})$.

To prove the increasing property, we need to take a look at what is
happening at $P_{ir}$. At $P=P_{ir}$, a boundary line $l_{ij*}$ shrinks
to $0$. After this shrink, $l_{ik^{*}}$,$l_{ih^{*}}$, which are
used to intersect with $l_{ij^{*}}$, start to intersect themselves, since
$j^{*}$ no longer being a neighbor of $i$.
The boundary line $\overline{\mathbf{x}_{uk}\mathbf{x}_{kj}\mathbf{x}_{jh}\mathbf{x}_{hv}}$
changes to $\overline{\mathbf{x'}_{uk}\mathbf{x}'_{kh}\mathbf{x}'_{hv}}$, as shown in Fig. \ref{fig:dldp}(a).

Consider the  intersection point $\mathbf{x}_{kh}$ of extended line. Since $P_{-i}$ is fixed and $P_{i}$ is increasing,  we know that  $\mathbf{x}_{kh}$ must be on the boundary line between $k^{*}$ and $h^{*}$,  and move towards the inside of convex polygon $\mathcal{S}'_{i}$ during the increment of $P_{i}$.
%Since the ploygen is convex, we  see from Fig. \ref{fig:dldp} that $\angle \overline{\mathbf{x'}_{uk}\mathbf{x}'_{kh}\mathbf{x}'_{hv}}<\pi$. \textcolor{red}{why need this??}

Now we study the change in length of these boundarys. Consider $\overline{\mathbf{x}_{uk}\mathbf{x}_{kh}}$ and  $\overline{\mathbf{x}'_{uk}\mathbf{x}'_{kh}}$. We see that they are parallel because the movement of boundary is perpendicular to the connecting line of $i,k$. Let us denote $\mathrm{d}\overline{\mathbf{x}_{uk}\mathbf{x}_{kh}}=\overline{\mathbf{x}'_{uk}\mathbf{x}'_{kh}}-\overline{\mathbf{x}_{uk}\mathbf{x}_{kh}}$
in  Fig. \ref{fig:dldp}, and set an axis parallel to describe both $\overline{\mathbf{x}_{uk}\mathbf{x}_{kh}}$ and $\overline{\mathbf{x}'_{uk}\mathbf{x}'_{kh}}$. Then, the 1D  coordinates for the endpoints $\mathbf{x}_{uk},\mathbf{x}'_{uk},\mathbf{x}_{kh},\mathbf{x}'_{kh},\mathbf{x}_{kj}$ on this axis will be  $x_{uk},x'_{uk},x_{kh},x'_{kh},x_{kj}$. Suppose we choose the direction of axis such that $x_{kh}>x_{kj}>x_{ku}$.  Then, we can express the change of length of $\overline{\mathbf{x}_{uk}\mathbf{x}_{kh}}$ in 1D as follows:
 \begin{eqnarray*}
\mathrm{d}\overline{\mathbf{x}_{uk}\mathbf{x}_{kh}}&&=\overline{\mathbf{x}'_{uk}\mathbf{x}'_{kh}}-\overline{\mathbf{x}_{uk}\mathbf{x}_{kh}}\\
&&=(x'_{kh}-x'_{uk})-(x_{kh}-x_{uk}) \\
&&=\mathrm{d}x_{kh}-\mathrm{d}x_{uk}
\end{eqnarray*}
Similarly, we have $\mathrm{d}\overline{\mathbf{x}_{uk}\mathbf{x}_{kj}}=\mathrm{d}x_{kj}-\mathrm{d}x_{uk}$.
Also, $l_{ik^*}=\overline{\mathbf{x}_{uk}\mathbf{x}_{kj}}$ becomes $l'_{ik^*}=\overline{\mathbf{x}'_{uk}\mathbf{x}'_{kh}}$. With the same axis, we have
$\mathrm{d}l_{ik*}=\mathrm{d}x_{kj}-\mathrm{d}x_{uk}=(x'_{kh}-x{}_{kj})-\mathrm{d}x_{uk}$.

From Fig. \ref{fig:dldp}, we  see that
\begin{eqnarray}
\overline{\mathbf{x}_{kj}\mathbf{x}_{kh}}&&=-\mathbf{d}\overline{\mathbf{x}_{kj}\mathbf{x}_{kh}} \nonumber\\
&&=-\mathrm{d}(\overline{\mathbf{x}_{uk}\mathbf{x}_{kh}}-l_{ik^*}) \nonumber\\
&&=\mathrm{d}x_{kj}-\mathrm{d}x_{kh} \nonumber\\
&&=x'_{kh}-x_{kj}-x'_{kh}+x_{kh} 	\nonumber\\
&&=x_{kh}-x_{kj}>0
\end{eqnarray}

Similarly, we have $\overline{\mathbf{x}_{kh}\mathbf{x}_{jh}}
=(-\mathrm{d}\overline{\mathbf{x}_{kh}\mathbf{x}_{hv}})-(-\mathrm{d}l_{ih^{*}})$. From the triangle inequality, we know that $\overline{\mathbf{x}_{kj}\mathbf{x}_{kh}}+\overline{\mathbf{x}_{kh}\mathbf{x}_{jh}}\geq\overline{\mathbf{x}_{kj}\mathbf{x}_{jh}}$, i.e.,
\begin{equation}
(-\mathrm{d}\overline{\mathbf{x}_{uk}\mathbf{x}_{kh}})-(-\mathrm{d}l_{ik^{*}})+(-\mathrm{d}\overline{\mathbf{x}_{kh}\mathbf{x}_{hv}})-(-l_{ih^{*}})\geq-\mathrm{d}l_{ij^{*}} \label{eq:triangle}
\end{equation}
Due to the continuity of  $\frac{\mathrm{d}^{2}W_{i}'}{\mathrm{d}P_{i}^{2}}$, we know that
\begin{equation*}
\frac{\mathrm{d}\overline{\mathbf{x}_{uk}\mathbf{x}_{kh}}}{\mathbf{d}P_i}|_{P_{ir}-\mathbf{d}P_{ir}\rightarrow P_{ir}}
=\frac{\mathrm{d}l'_{ik^{*}}}{\mathrm{d}P_{i}}|_{P_{i}=P_{ir}},
\end{equation*}
and that
\begin{equation*}
\frac{\mathrm{d}\overline{\mathbf{x}_{kh}\mathbf{x}_{hv}}}{\mathbf{d}P_i}|_{P_{ir}-\mathbf{d}P_{ir}\rightarrow P_{ir}}
=\frac{\mathrm{d}l'_{ih^{*}}}{\mathrm{d}P_{i}}|_{P_{i}=P_{ir}}.
\end{equation*}
Thus,  we have from (\ref{eq:triangle}) that:
\begin{eqnarray}
&&\frac{\mathrm{d}^{2}W_{i}'}{\mathrm{d}P_{i}^{2}}|_{P_{i}=
P_{ir}}-\frac{\mathrm{d}^{2}W_{i}}{\mathrm{d}P_{i}^{2}}|_{P_{i}=P_{ir}}	\nonumber\\
&&=\frac{\mathrm{d}^{2}W_{ik^{*}}}{\mathrm{d}P_{i}^{2}}|_{P_{i}\rightarrow P_{ir}-}+\frac{\mathrm{d}^{2}W_{ih^{*}}}{\mathrm{d}P_{i}^{2}}|_{P_{i}\rightarrow P_{ir}-} \nonumber\\
&&\quad-(\frac{\mathrm{d}^{2}W_{ik^{*}}}{\mathrm{d}P_{i}^{2}}|_{P_{i}\rightarrow P_{ir}+}\frac{\mathrm{d}^{2}W_{ih^{*}}}{\mathrm{d}P_{i}^{2}}|_{P_{i}\rightarrow P_{ir}+}+\frac{\mathrm{d}^{2}W_{ij^{*}}}{\mathrm{d}P_{i}^{2}}|_{P_{i}\rightarrow P_{ir}+}) 	\nonumber\\
&&=\frac{\mathrm{d}l_{ik^{*}}}{\mathrm{d}P_{i}}+\frac{\mathrm{d}l_{ih^{*}}}{\mathrm{d}P_{i}}+\frac{\mathrm{d}l_{ij^{*}}}{\mathrm{d}P_{i}}-\frac{\mathrm{d}\overline{\mathbf{x}_{uk}\mathbf{x}_{kh}}}{\mathrm{d}P_{i}}-\frac{\mathrm{d}\overline{\mathbf{x}_{kh}\mathbf{x}_{hv}}}{\mathrm{d}P_{i}}\geq0.
\end{eqnarray}

Combining  these results, we have that $\frac{\mathrm{d}^{2}W_{i}}{\mathrm{d}P_{i}^{2}}$
is a piecewise increasing function, and each piece is a fragment
of an increasing downward parabola.
%Noted that situation in Fig. \ref{fig:dldp} (b) also leads to the same result.

Now we  confirm the quasiconcavity of $W_i(P_i)$. We first have $W(P_{i})\geq0$ for $P_{i}\in[0,\infty)$.
%Hence if $W(0)=0$ and $W(P_{max})=0$, according to the Mean Value Theorem, there exists at least one $P_{i}^{*}$ such that $\frac{\mathrm{d}W_{i}}{\mathrm{d}P_{i}}|_{P_{i}=P^{*}}=0$.
If $\frac{\mathrm{d}^{2}W_{i}}{\mathrm{d}P_{i}^{2}}<0$ in $[0, P_{upper}]$, since $W_i(P_i)$ is continuous, we know $W_i(P_i)$ is quasiconcave.
%
%$P^{*}\in(P',P_{max}]$
Otherwise we consider the case when $P_{i,max}$ exists.
%let  $P'$ denote the price  such that   at $P_i<P'$ we have
% $\frac{\mathrm{d}^{2}W_{i}}{\mathrm{d}P_{i}^{2}}<0$ and at $P_i>P'$ have $\frac{\mathrm{d}^{2}W_{i}}{\mathrm{d}P_{i}^{2}}>0$.
 %\textcolor{red}{why exists?}
Since $\frac{\mathrm{d}W_{i}}{\mathrm{d}P_{i}}(P_{i,max})<0$,  $\frac{\mathrm{d}W_{i}}{\mathrm{d}P_{i}}(0)>0$, and we know that $\frac{\mathrm{d}^{2}W_{i}}{\mathrm{d}P_{i}^{2}}$ is monotone,  $\frac{\mathrm{d}W_{i}}{\mathrm{d}P_{i}}$ will be non-decreasing before maximum and then non-increasing until $P_i=P_{i,max}$, which means $\frac{\mathrm{d}W_{i}}{\mathrm{d}P_{i}}$ is quasiconcave.

If $P_{i,max}$ does not exists,  $W(P_{upper})>0$. Yet we know that the left side derivative of $W_i(P_i)$ will be negative at $P_{upper}$, since $P_{upper}$ is   high enough such that increment in price only leads to market area decrease. Similarly we can prove the quasiconcavity of $W_i(P_i)$.

Since $W(P_{i})$ is continuous in $P_{j},\forall j\in\mathcal{NB}(i)$
and quasiconcave, from Theorem \ref{thm:1952}, Nash equilibrium always exists. %\textcolor{red}{we should state this as a theorem up front. then use it.}
\end{proof}

\subsection{Proof of Theorem \ref{thm:2d-condition}}
\begin{proof}
Based on the discussion in the proof of Theorem \ref{theorem:2d-exist}, we know that when market is  at equilibrium,
\begin{eqnarray*}
\frac{\mathrm{d}W_{i}}{\mathrm{d}P_{i}}=S_{i}+P_{i}\frac{\mathrm{d}S_{i}}{\mathrm{d}P_{i}}=0
\end{eqnarray*}
and
\begin{eqnarray*}
\frac{\mathrm{d}S_{i}}{\mathrm{d}P_{i}}=\sum_{j\in\mathcal{NB}(i)}-\frac{l_{ij}}{2d_{ij}},
\end{eqnarray*}
which give us the theorem:
\begin{eqnarray*}
P_i=\frac{1}{\sum_{j\in\mathcal{NR}(i)}\frac{l_{ij}}{2d_{ij}}}S_i, \quad\forall i\in \mathcal{N}.
\end{eqnarray*}
\end{proof}

\subsection{Proof of Theorem \ref{thm:1d-q1-exist}}
To prove Theorem \ref{thm:1d-q1-exist}, we need the following lemmas, whose proof will be given after the proof.
\begin{lemma}\label{lemma:potential}
Company $i$'s potential competitor $j$ provides the same aggregate price as $i$  only at the boundary of $i$.
\end{lemma}

\begin{proof} (Theorem \ref{thm:1d-q1-exist})
We first consider the situation when no ``wipe out''  happens in the market, that is,
\begin{equation}
\beta<\frac{2d_{i}d_{i-1}}{d_i+d_{i-1}}\quad \forall i\in\mathcal{N}.	\label{condition:wipe-out}
\end{equation}

Based on the discussion in Section 6.1,  we know that $S_i(P_i)$ is piecewise continuous, we still need to prove that $S_i(P_i)$ is continuous at the junction point between two pieces.  Without loss of generality,  suppose $i$'s right neighbor changes from $j$ to $u$ after the junction point while left neighbor remains $k$.  From Lemma \ref{lemma:neighbor} we know that $x_i<x_u<x_j$.

According to Lemma \ref{lemma:potential}, we know that at the junction point of $S_i$,  there must exists a potential competitor $u$ provides the same aggregate price at $i$ and $j$'s boundary point.  Otherwise if  $i$ has no potential competitors, we can show that $i$'s neigbors will not change.
 This can be explained in the following two cases :(i) If $i$ increases an infinitely small price, since $i$ has no potential competitors, then customers at border of $i$ will immediately  choose only from current neighbors of $i$. Hence $i$'s neighbor will have an increment in market area, and no new neighbor of $i$ appears during price increment. (ii) If $i$ lower its price by an infinitely small amount, each neighbor of $i$ will suffer from a market area loss. We still need to say that during this process, no neighbor of $i$ will disappear, as long as the price decrease of $i$ is small enough. This is because for any neighbor of $i$, say $j$, though it suffers from a market area decrease, yet its neighbor $t\in \mathcal{NR}(j), t\neq i$ will increase its market area and provides lower price at $j$ and  $t$'s old boundary(the boundary between $j$ and $k$ before $i$ decreases its price), thus no potential competitor of $j$ can start to survive. In this case, we rule out the probability that $j$ will be wiped out.

 After a sufficiently small increase of $P_i$, $u$ begins to survive and gradually increase its market area starting from zero, which means $i$ do not suffer any sudden drop in market share. After then, $i$'s neighbor change into $k$ and $u$, and $\mathrm{d}S_i/\mathrm{d}P_i$ also change. So far we have shown that $S_i(P_i)$ is continuous and piece-wise differentiable (the derivative changes when the neighbors change), so as $W_i(P_i)=P_iS_i$.

Each continuous piece of $S_i(P_i)$ is a concave function. That is because the growth rate of $\frac{\mathrm{d}S_{i}}{\mathrm{d}P_{i}}$ will always be non-positive. To understand this， we can consider the following fact. Let $\delta>0$ be a sufficiently small value, we will be able to conclude that in each continuous piece of $S_i(P_i)$ we have  $0>\frac{\mathrm{d}S_{i}}{\mathrm{d}P_{i}}|_{P_i-\delta}>\frac{\mathrm{d}S_{i}}{\mathrm{d}P_{i}}|_{P_i+\delta}$.
The reason is  $i$'s market area become smaller at $P_i+\delta$ compare to $P_i-\delta$, while both its neighbors  enjoy market area increment during the same time, which will lead to more market area shrink for $i$ at $P_i+\delta$ than at $P_i-\delta$, since one's market area has positive feedback when $q=1$.

 We now prove the fact that at any junction point (a break point between two pieces),  $P_{junc}$ of $W_i$, if the left derivative of $W_i(P_{junc})$ is negative, then the right derivative of $W_i(P_{junc})$ will also be negative. Thus, $W$ will monotonically increase to a maximum point and then monotonically decrease, which is quasiconcave. Since $\frac{\mathrm{d}W_{i}}{\mathrm{d}P_{i}}=S_i+P_i\frac{\mathrm{d}S_{i}}{\mathrm{d}P_{i}}$, we only need to consider the change in $\frac{\mathrm{d}S_{i}}{\mathrm{d}P_{i}}$ at $P_i=P_{junc}$. Let  $\frac{\mathrm{d}S_{i}}{\mathrm{d}P_{i}}|_{P_{junc}+}$ denote the right derivative of $S_i$ at $P_{junc}$ and $\frac{\mathrm{d}S_{i}}{\mathrm{d}P_{i}}|_{P_{junc}-}$ for the left derivative. By proving $\frac{\mathrm{d}S_{i}}{\mathrm{d}P_{i}}|_{P_{junc}+}<\frac{\mathrm{d}S_{i}}{\mathrm{d}P_{i}}|_{P_{junc}-}$ we can achieve our goal.

%Suppose the contradict that $\frac{\mathrm{d}W_{i}}{\mathrm{d}P_{i}}|_{P_{junc}+}>\frac{\mathrm{d}W_{i}}{\mathrm{d}P_{i}}|_{P_{junc}-}$ at $P_{junc}$.

At $P_i=P_{junc}$, we have $P_i\{R_i\}=P_{j}\{R_i\}=P_{u}\{R_i\}$. Suppose $u$ does not exist and  $P_i$ has a positive deviation $\mathrm{d}P_i\rightarrow0$, we can represent the price at $i$'s new boundary $R_i'$ as $P_i'\{R_i'\}=P_j'\{R_i'\}$. We can see that $P_i'\{R_i'\}=P_i\{R_i\}+\mathrm{d}P_i+2\mathrm{d}R_i(R_i-x_i)+(\mathrm{d}R_i)^2-\beta\mathrm{d}S_i$.
Consider the fact that customers at $i$ and $j$'s boundary suffers the highest price from $i$ or $j$, because they are farthest to their suppliers. Hence, if during the $\mathrm{d}P_i$ increment, $P_u\{R'\}>P'_i\{R'\}$, then no customers will choose to buy from $u$, thus $u$ will not survive and $j$ will remain $i$'s neighbor. Since we consider the situation when $u$ starts to survive, we know that $P_u\{R'\}<P'_i\{R'\}$, and that right boundary of $i$ has shrunk more compare to the situation when $j$ being its neighbor, so does the left boundary of $i$. Thus, $\frac{\mathrm{d}S_{i}}{\mathrm{d}P_{i}}|_{P_{junc}+}<\frac{\mathrm{d}S_{i}}{\mathrm{d}P_{i}}|_{P_{junc}-}$.

Hence, $W_i(P_{i})$ is a quasiconcave function and $W_i(P_{i})=0$. If there exists a positive price such that $W_i(P_i)=0$, denote it as $P_{i,max}$. We can see that for $P_i\geq P_{i,max}$, $W_i(P_i)=0$.

According to Theorem \ref{thm:1952}, since  $W_{i}(\mathbf{P})$ is continuous in $\mathbf{P}\in [0, P_{upper}]^N$ and quasiconcave
in $P_{i}$,  Nash equilibrium  exists when no ``wipe out'' can happens.

For the case when  condition (\ref{condition:wipe-out}) does not hold, we can operate in the following way to give at least one equilibrium. We choose some companies to be hidden, which means they can be seen as not existed in the market, while the others being ``activated.''  For the simplicity in presentation below, we define a condition among those activated companies:
\begin{eqnarray}\beta<\frac{2d_{i}d_{i-1}}{d_i+d_{i-1}}\quad \forall\,\, \text{activated company } i .	
\label{condition:activated-wipe-out}
\end{eqnarray}
  The way we choose is as follow:
(i) Hide all companies first. Then activate companies one by one until we can not activate any one more companies without bringing ``wipe out'' phenomenon into the market. That is, all activated companies satisfied the condition (\ref{condition:activated-wipe-out}), but activating any one of the hidden companies will violates it.
(ii) Then we will check that when any one hidden company, say $j$, is activated, it violates the inequality $\beta<\frac{2d_{j}d_{j-1}}{d_j+d_{j-1}}$. Otherwise according to (i), there must exists one or two of its neighbors, say $k$, such that $\beta>\frac{2d_{k}d_{k-1}}{d_k+d_{k-1}}$. Then we activate $j$ and hide $k$.
(iii) Now we can guarantee that all activated companies satisfy the condition  (\ref{condition:activated-wipe-out}), while each  hidden company is being ``wiped out''.

Then for activated companies there exists equilibrium as proved before. Now let those hidden companies come back to the market with price $P_{upper}$. From the analysis of equation (15), we can see that hidden company $i$ will never survive by unilateral price change. For any activated company $j$, since hidden companies are at the status of price being  $P_{upper}$ and have no market area, they can never survive by $j$'s unilateral action. Hence, $j$'s price will remain unchanged. This setting guarantees at least one equilibrium when (\ref{condition:wipe-out}) does not hold.

Concluding all the results above, we have proved the existence of Nash equilibrium in 1D market with $q=1$.
\end{proof}

\subsection{Proof of Lemma \ref{lemma:potential}}
\begin{proof}
Suppose $i$'s  potential competitor $j$ satisfy the condition (1) in Definition 6. Since $S_j=0$,  we can show that $j$ only has at most one point of market area.
This is because we can use similar method in  Theorem \ref{thm:convex} to show that $j$'s market point set is convex, i.e., line segment in 1D market,  and since $S_j=0$, $j$ has only one market point.
Since price on the boundary of $i$ is highest price among all points in $\mathbf{cl}(\mathcal{S}_i)$,  $j$ can only provide the same price at a boundary point of $i$. Otherwise we can always find a nonzero neighborhood of $\mathcal{V}(\mathbf{x})$, where $P_j\{\mathbf{x}\}=P_i\{\mathbf{x}\}$, $\mathbf{x}\in\mathbf{cl}(\mathcal{S}_i)$, such that $\exists \mathbf{x}'\in \mathcal{V}(x), P_j\{\mathbf{x}'\}<P_i\{\mathbf{x}'\}$, which violates $S_j=0$.

Since we consider the 1D market, $i$'s  potential competitor $j$ will never satisfy the condition (2) in Definition 6. This is because $i$ only has two boundary points in 1D market, for any neighbor $k$ of $i$, we must have $\frac{|\mathbf{cl}(\mathcal{S}_k)\cap\mathbf{cl}(\mathcal{S}_i)|}{|\mathcal{BR}(i)|}=\frac{1}{2}$.
\end{proof}

\subsection{Proof of Theorem \ref{thm:condition}}
\begin{theorem}\label{thm:condition}
When a market with dimension $K$ is at a Nash equilibrium, for any company $i$, if it has no potential competitors,
\begin{equation} \label{eq:1d-part1}
P_i^*=c_i\cdot S_i^*,  \quad\quad \forall i \in \mathcal{N}_S
\end{equation}
where $c_i=-\frac{dP_{i}}{dS_{i}}|_{P_{i}=P_{i}^{*}}>0$.

Otherwise, we have:
\begin{equation} \label{ineq:equil}
c_i'\cdot S_{i}^*\leq P_{i}^*\leq c_i''\cdot S_{i}^*, \quad\quad \forall i \in \mathcal{N}_S.
\end{equation}
Here $c_i'=-\frac{dP_{i}}{dS_{i}}|_{P_{i}=P_{i}^{*}+}>0$, $c_i''=-\frac{dP_{i}}{dS_{i}}|_{P_{i}=P_{i}^{*}-}>0$.
Inequality (\ref{ineq:equil}) becomes an equality only when $i$'s potential competitor does not survive  at $P_i^*$.
\end{theorem}
We know from Theorem \ref{thm:1d-q1-exist} that $W_i(P_i)$ is continuous.

First we prove the result when $i$ has no potential competitors at equilibrium. This condition guarantees that with an infinitely small deviation of price of company $i$, its neighbors will remain the same, as explained in the proof of Theorem \ref{thm:1d-q1-exist}, which also ensures that  the payoff function is continuous and differentiable here.

Note that the analysis above is insensitive to $q$ and market dimension $K$. Hence we have for $i$ at equilibrium,
\begin{eqnarray}\label{eq:thm-proof-dw}
\frac{\mathrm{d}W_{i}}{\mathrm{d}P_{i}}=S_{i}+P_{i}\frac{\mathrm{d}S_{i}}{\mathrm{d}P_{i}}=0.
\end{eqnarray}
Thus, we prove the first equation (\ref{eq:1d-part1}) of the theorem.
\begin{eqnarray*}
P_i^*=c_i\cdot S_i^*,  \quad\quad \forall i \in \mathcal{N}_S
\end{eqnarray*}
where $c_i=-\frac{dP_{i}}{dS_{i}}|_{P_{i}=P_{i}^{*}}>0$.

For the situation when $i$ has potential competitors at equilibrium, $W_i(P_i^*)$ may not be differentiable here since $S_i(P_i)$ changes after the potential competitors survives. We thus consider the left and right derivatives at $P_i^*$.
Since the market is at the equilibrium, any unilateral deviation of $P_i^*$ will decrease $i$'s payoff. Hence,  we must have:
\begin{eqnarray}
\frac{\mathrm{d}W_i}{\mathrm{d}P_i}|_{P_i^*-}>0	,\quad
\frac{\mathrm{d}W_i}{\mathrm{d}P_i}|_{P_i^*+}<0 \label{ineq:left-right}
\end{eqnarray}
These two inequalities imply that $W_i$ reaches its maximum at $P_i^*$, i.e., $W_i$ increases before $P_i^*$ and decreases after it. Plugging (\ref{ineq:left-right}) into (\ref{eq:thm-proof-dw}) we prove  inequality (\ref{ineq:equil}) of the theorem.
It can be seen that if $i$'s potential competitors did not survive at $P_i^*$, the form of $S_i$ will not change. Hence, $W_i$ will be differentiable at $P_i^*$,  and  inequality (\ref{ineq:equil}) becomes equation (\ref{eq:1d-part1}).

Summing results above proves the theorem.

\bibliographystyle{abbrv}
\bibliography{acmsmall-market-bibfile}

%
% For AAMAS-2016, as references are unlimited but appendices must fit within
% 8 pages, the References section must come after the appendices (if any)
%\bibliographystyle{abbrv}
%\bibliography{sigproc}  % sigproc.bib is the name of the Bibliography in this case
% ACM needs 'a single self-contained file'!
%\balancecolumns % GM June 2007
% That's all folks!
\end{document}